\documentclass[12pt]{article}

\usepackage{amsmath}
\usepackage[utf8x]{inputenc}
\usepackage[english]{babel}
\usepackage[top=1in,bottom=1in,left=1in,right=1in]{geometry}
\usepackage{graphicx}
\usepackage[normalem]{ulem}
\usepackage{amsmath, amsthm}
\usepackage{amssymb}
\usepackage{amsthm}
\usepackage{xcolor}
\usepackage{soul}
\usepackage{enumitem}
\usepackage{ragged2e}
\usepackage{nomencl}
\usepackage{soul}
\usepackage{natbib}
\usepackage{tikz}
\usepackage{hyperref}
\hypersetup{colorlinks=true, linkcolor=blue, citecolor=blue}

\newtheorem{definition}{Definition}
\newtheorem{remark}{Remark}
\newtheorem{theorem}{Theorem}
\newtheorem{corollary}{Corollary}
\newtheorem{proposition}{Proposition}
\newtheorem{example}{Example}
\newtheorem{observation}{Observation}
\newtheorem{lemma}{Lemma}
\newtheorem{assumption}{Assumption}

\newcommand{\bq}{\mathbf{q}}
\newcommand{\br}{\mathbf{r}}
\newcommand{\R}{\mathbb{R}}
\newcommand{\E}{\mathbb{E}}
\newcommand{\N}{\mathbb{N}}
\newcommand{\oU}{\overline{U}}
\newcommand{\bo}{\mathbf{o}}

\DeclareMathOperator{\val}{Val}

\DeclareMathOperator{\inter}{int}
\DeclareMathOperator{\sign}{sign}
\DeclareMathOperator{\best}{br}
\DeclareMathOperator{\emp}{emp}
\DeclareMathOperator{\Lim}{Lim}
\DeclareMathOperator{\lex}{lex}

\title{Testing Decision Makers without Counterfactuals}

\author{Yakov Babichenko\footnote{Technion - Israel Institute of Technology. Email: yakovbab@technion.ac.il.}}
\date{}

\begin{document}
\maketitle

\begin{abstract}

A decision-maker (DM) repeatedly makes choices under uncertainty in a bandit environment, where only the realization of the chosen arm is observed. Another competing agent, the adviser (AD), repeatedly provides recommendations, but the realizations of these recommendations are unobserved unless they coincide with the DM’s choice. Both agents possess partial information about the arms’ realizations. The central question we focus on is whether, in the long run, an outside observer can identify which agent is more informed based solely on the observed decisions, recommendations, and arm realizations.

A test selects one of the agents based on the observed data. We focus primarily on the class of \emph{scoring tests}, which assign a numerical score to each observation and select the agent according to the average score. We study strategic agents whose objective is to be selected by the test. For simultaneous arm choices, we show that there exists a scoring test that successfully identifies the more-informed agent. For sequential arm choices, however, no such scoring test exists.

Finally, we explore the tension between identifying the more-informed agent and maximizing welfare. A DM whose objective is to pass the test may not necessarily make welfare-maximizing decisions. In a binary-arm environment, we show that no scoring test can simultaneously identify the more informed agent and achieve more than half of the welfare attained by welfare-maximizing decisions.

\end{abstract}

\section{Introduction}

\emph{''Counterfactual reasoning, which deals with what-ifs, might strike some readers as unscientific. Indeed, empirical observation can never confirm or refute the answers to such questions.”} 

\hspace*{\fill} Judea Pearl, \emph{The Book of Why: The New Science of Cause and Effect}

\vspace{3mm}

The concern noted by Judea Pearl regarding unobserved counterfactuals is particularly relevant when evaluating the quality of decision-making. If a decision leads to poor consequences, one could always argue that an alternative choice would have produced even worse consequences; a claim that cannot be confirmed or refuted through empirical observation.

Many decision problems involve counterfactuals that are difficult, if not impossible, to evaluate, such as geopolitical choices, corporate strategies, or major life decisions. Yet in these situations, it remains desirable to assess the quality of decisions, ideally relying only on empirical observations. As a leading motivation, we focus on geopolitical decisions, though the model we propose applies more broadly to any setting where one agent currently makes decisions on behalf of an entity and another competing agent could potentially replace the current decision-maker in the future.

Geopolitical decisions are made by the \emph{government} (the \emph{decision-maker} in our model). To evaluate these decisions, one needs a reference point for comparison. A natural reference point, often available in such contexts, is the \emph{opposition} (the \emph{adviser} in our model), which proposes alternative courses of action.\footnote{Other potential reference points include the performance of comparable countries or of the same country in the past. These, however, are debatable, as it is often unclear which country or which historical period is most comparable to the present situation.} There is, however, an inherent asymmetry between the government and the opposition: the consequences of the government’s decisions are observable, while the opposition’s proposals remain counterfactuals whose consequences will never be observed.

Decisions of this kind are typically made under uncertainty. The quality of an agent (the government or the opposition) is determined by the information it holds about the decision problem. More concretely, an agent must decide which arm to pull and privately holds information about the possible realizations. Only the realization of the arm pulled by the government is observed. We pose the following question:

\begin{center}
\emph{If one of these agents is more informed about the arms’ realizations than the other, can an outside observer identify which agent it is?}
\end{center}

Specifically, by \emph{more informed} we mean that one agent has access to the other’s information and, in addition, possesses private information of her own; see \citep{blackwell1953equivalent}. In the context of geopolitical decisions, this notion of more informativeness may capture a more accurate perception of reality. By \emph{identifying} the more informed agent, we mean selecting one of the agents based solely on observable data. In our case, the observable data are given by the government’s decisions, their realized welfare, and the opposition’s advice.

We consider the above problem in a repeated environment, where the outside observer has access to a sequence of such observations. Our question is whether it is possible, in the long run, to identify the more-informed agent. A selection rule—namely, a mapping that assigns an agent to every sequence of observations—will be called a \emph{test}. In the context of geopolitical decisions, a test can be thought of as a guideline for the public to assess the relative quality of the government and the opposition.

However, the fact that both agents are being tested can affect their behavior. Each seeks to ensure that, by the end of the interaction, her track record appears favorable. We take an extreme perspective and assume that each agent’s sole objective is to maximize her chances of being selected by the test. This assumption induces a zero-sum game between the agents. The question posed above can thus be reinterpreted as whether there exists a test that identifies the more-informed agent in equilibrium of the zero-sum game induced by the test. Whenever such a test exists, it naturally raises a further question:

\begin{center}
\emph{What behavior does such a test impose on the decision-maker? And how far does she deviate from the welfare-maximizing course of action?}
\end{center}

The public, in fact, has two objectives: to evaluate the quality of the decision-maker (to select a capable leader in the future) and, at the same time, to incentivize the decision-maker to choose optimal actions (to secure good welfare in the present). If the test induces substantially suboptimal decisions, however, the effectiveness of such an assessment becomes questionable.

\subsection{Main Results}

We primarily focus on the class of \emph{scoring tests}, which assign a numerical value to each observation (independently of time) and determine the winner by averaging these values.\footnote{Equivalently, a scoring test can be interpreted as assigning a score to each agent, with the winner being the agent whose average score is higher. Focusing instead on the difference between the two scores yields an equivalent representation by a single numerical value.} This class is natural in that it treats every period symmetrically. We distinguish between \emph{simultaneous decisions}, where the two agents choose an arm at the same time, and \emph{sequential decisions}, where one agent reacts to the choice of the other. In the geopolitical context, the sequential-decision setting is arguably more natural: typically, the opposition responds to the decisions made by the government rather than announcing a course of action simultaneously.
Our main results are as follows:

\begin{enumerate}
    \item There exists a scoring test that identifies the more informed agent in the simultaneous-decision environment. This holds even when the environment changes over time—that is, when each period corresponds to a different bandit instance (see Theorem \ref{theo:simul}).
    \item No scoring test exists that can identify the more informed agent in the sequential-decision environment. This remains true even when the environment is fixed over time—that is, when the same bandit instance is repeated in every period (see Theorem \ref{theo:seq}).
\end{enumerate}

Regarding the interplay between identification and the welfare induced for the public, we obtain the following negative result:

\begin{enumerate}
\setcounter{enumi}{2}
\item In the binary-arm environment, no scoring test exists that both identifies the more informed agent and achieves more than half of the optimal welfare (see Theorem \ref{theo:welfare}).
\end{enumerate}

Here, \emph{optimal} welfare refers to the realizations generated by welfare-maximizing decisions of the decision-maker, rather than by choices aimed at winning the test. Note that a one-half approximation is an extremely low benchmark, as it can already be attained by choosing an arm uniformly at random, without using any information. Thus, even though it is possible to identify the more informed agent, doing so inevitably creates incentives for the decision-maker to act in a substantially suboptimal way.

We further explore in detail a particularly natural special instance of our setting in which there is a single safe arm and a single risky arm. This special case admits a natural interpretation: in every period, a risky project arrives, and the DM chooses whether to accept or reject it. The positive result for simultaneous decisions (Theorem \ref{theo:simul}) clearly applies in this setting. Likewise, the negative result concerning the interplay between identification and welfare under simultaneous decisions (Theorem \ref{theo:welfare}) continues to hold, since its proof is explicitly based on this special case.

In contrast, the impossibility result for identifying the more informed agent in a sequential decision environment does not extend to this setting. For sequential decisions in this special case, we show that identification is possible when the AD acts first (Proposition \ref{pro:ad1}), but impossible when the DM acts first (Proposition \ref{pro:dm1}).

\paragraph{\textbf{A high-level intuition of the proof techniques.}}
The main challenges in analyzing these repeated zero-sum interactions are twofold. First, repeated interactions are typically harder to analyze than one-shot interactions. Second, each period involves a game with incomplete information, whose analysis is typically intricate.

The difficulty of repeated interaction can be addressed fairly directly. The additive structure of \emph{scoring tests} creates an equivalence between the one-shot and repeated settings; see Lemmas~\ref{lem1} and~\ref{lem2}. By contrast, handling the incomplete information challenge is more subtle. To address this, we focus on the class of \emph{complete}-information games $\{G(s,p)\}$, where $s$ is a score and $p$ is a \emph{common} prior on arms’ realizations shared by both agents. The decision-maker aims to maximize the score, while the adviser aims to minimize it.

It turns out that the class $\{G(s,p)\}_p$ is sufficiently expressive to determine whether a scoring function $s$ succeeds in identifying the more-informed agent; see Lemma~\ref{lem:equiv}. Conceptually, Lemma~\ref{lem:equiv} is a significant insight: to prove (or disprove) that the identification property holds for a class of incomplete-information games, it suffices to restrict attention to \emph{complete}-information games and establish (or refute) a different property there. This key observation allows us to prove our main results without analyzing games with incomplete information directly. Moreover, bounding welfare in an equilibrium (Theorem~\ref{theo:welfare}) can likewise be done by focusing solely on the class $\{G(s,p)\}_p$, leading to sharp upper bounds.

\subsection{Related Literature}

This paper is conceptually related to the literature on \emph{expert testing}, which examines whether one can evaluate the quality of \emph{probabilistic forecasts}, rather than decisions, provided by an expert. A survey of this literature appears in \citep{Olszewski2011Calibration}. A surprising result in this literature is that it is impossible to distinguish between an uninformed expert and an expert who knows the true distribution; see \citep{Sandroni2003Reproducible,DekelFeinberg2006,OlszewskiSandroni2008Econometrica,OlszewskiSandroni2009AOS_Manipulability}. However, by slightly relaxing the testing requirements \citep{OlszewskiSandroni2009AOS_Nonmanipulable} or by imposing additional structural assumptions \citep{SandroniShmaya2013}, such a distinction becomes possible.

Our problem is closer to the \emph{comparative} variant of expert testing, where the goal is to determine which of two experts is better informed. In this setting, identification of the more-informed expert is possible; see \citep{AlNajjarWeinstein2008,KavalerSmorodinsky2019}. Moreover, such tests typically incentivize experts to report their most accurate forecasts.
In contrast, our results show that this desirable property does not extend to decision-making without counterfactuals. Although it is possible to identify which decision-maker is better (Theorem~\ref{theo:simul}), doing so necessarily comes at the cost of decision quality (Theorem~\ref{theo:welfare}).

A key construct in this literature is the zero-sum game between nature and the expert, in which nature acts adversarially. Our analysis also employs this construct in the proof of Theorem~\ref{theo:seq}.

The question of which counterfactuals can be inferred from observed data, and which cannot, has been extensively studied from a methodological perspective; see, for example, \citep{ShpitserPearl2012,BalkePearl1994,Pearl2009Causality}. Roughly speaking, the counterfactual whose validity we aim to test is the following: \emph{Had we followed the adviser’s recommendations, would our situation have been better?}

The main distinction between our paper and this line of research lies in the strategic nature of our model. While that literature focuses on data generated by a non-strategic entity (i.e., nature), in our setting the observed data are shaped by the behavior of strategic agents. In particular, the test we design influences the agents’ decisions and, consequently, the data that are observed.

A non-strategic analogue of our problem would assume that agents simply choose the best decision according to their available information. In such a model, it is \emph{impossible} to determine from the observed data whether following the adviser would have resulted in higher welfare; see Section \ref{sec:wel-max-agents}. In other words, the counterfactual stated above cannot be tested.

Bandit problems with a changing environment, as considered in our model, were introduced by \citet{garivier2008upper}, but in those settings the arms’ distributions evolve gradually over time, unlike in our case. The broader bandit literature primarily focuses on achieving \emph{low regret} \citep{lai1985asymptotically}, that is, performing (almost) as well as the best fixed arm in hindsight. This objective, however, is problematic in our context for two reasons. First, regret cannot be inferred from empirical observations. Second, in changing environments where even the set of feasible actions may vary over time, the notion of regret becomes conceptually meaningless.

Classical phenomena such as the trade-off between exploration and exploitation \citep{thompson1933likelihood,robbins1952some} or exploitation free-riding \citep{bolton1999strategic,kremer2014implementing} do not arise in our setting, since agents hold Bayesian beliefs about realizations that are independent across periods. Nevertheless, exposure to additional information about the arms by observing the opponent's action remains a crucial feature of the sequential-decision environment we study; see, for example, Proposition~\ref{pro:noTest}.

\section{Model}\label{sec:model}

We consider a repeated bandit environment with $n$ arms, denoted by ${1,2,\ldots,n}=[n]$. In the context of agents choosing an arm, the arms will also be referred to as \emph{actions}. The interaction takes place in discrete periods $t \in \mathbb{N}$.

\begin{remark}
All our negative results remain valid in an infinite-horizon environment, where the test has access to more information. For the positive results, we provide a finite-horizon analog.
\end{remark}

We begin by describing the distribution of arm realizations in a single period $t \in \N$, and, for clarity of exposition, we omit the index $t$.

We assume that the realization of each arm $i\in [n]$ belongs to a finite set $U_i$. If arm $i$ is chosen, the realization $u_i \in U_i$ is interpreted as \emph{welfare}. In particular, we focus on a common-interest environment in which the society shares a common utility function. Without loss of generality, we assume $U_i \subset [1,K]$ for some constant $K>1$, since any linear transformation of arm welfare leaves the model unchanged.

We denote by $U = U_1 \times \cdots \times U_n$ the set of welfare profiles across all arms, and by $\oU = \cup_{i \in [n]} U_i$ the set of all possible realizations. An arbitrary realization profile (one realization for each arm) is written as $u = (u_1, \ldots, u_n) \in U$. The vector of random realizations is $X = (X_1, \ldots, X_n) \in \Delta(U)$. The probability distribution of $X$ is denoted by $p \in \mathbb{R}^U$, where $p(u) = \mathbb{P}[X = u]$.
\begin{remark}
    All our positive results apply to the general case, where arbitrary correlations between arms are allowed (as assumed above). All our negative results hold even in the restricted setting where the realizations of the arms are independent.
\end{remark}

\paragraph{\textbf{Information about $X$.}} Our model involves two agents who have partial information about the realizations, that is, partial information about $X$. Moreover, one of the agents is more informed than the other, in the sense of \citep{blackwell1953equivalent}. Formally, this can be represented through signals: the less informed agent observes a signal $\sigma_1$ (a random mapping from states $U$ to an abstract set of messages), while the more informed agent observes a pair of signals $(\sigma_1, \sigma_2)$. An equivalent and, for our purposes, more convenient way to represent information is through the posteriors induced by these signals.

Let $\bq = (q^1, \ldots, q^l)$ be a collection of $|U|$-dimensional vectors such that 
$\sum_{j \in [l]} q^j=p$  and  $q^j\geq 0$  for every  $j\in [l]$.
We interpret $(q^j)_{j \in [l]}$ as partial information about $p$, where $j \in [l]$ represents the possible signal realizations. The realization $j$ occurs with probability $||q^j||_1$, and the posterior after observing signal $j$ is given by $\frac{q^j}{||q^j||_1}$. Thus, the collection $\bq$ captures the posteriors (together with their probabilities) of the less informed agent.

Recursively, each $q^j$ can be further decomposed as $q^j = \sum_{j' \in [l]} r^{j,j'}$ with $r^{j,j'} \geq 0$, to capture the posteriors of the more informed agent. Let $\br = (r^{j,j'})_{j,j' \in [l]}$ denote the collection of vectors corresponding to the more informed agent. This agent holds the posterior $\tfrac{r^{j,j'}}{||r^{j,j'}||_1}$ with probability $||r^{j,j'}||_1$ for each $j, j' \in [l]$. Thus, $\br$ represents the distribution over posteriors of the more informed agent. 

To summarize, the \emph{information structure} in our setting is described by the triple $I = (p, \bq, \br)$, where $p$ is the prior distribution of arm realizations, $\bq$ represents the distribution over posteriors of the less informed agent, and $\br$ represents the distribution over posteriors of the more informed agent. Observe that any information structure with finitely many signal realizations can be expressed in this form, where $l$ denotes the maximal number of signals.\footnote{The fact that $l$ is common to all decompositions is without loss of generality, since $0$-vectors can be included in the split.}

\paragraph{\textbf{Marginals.}} Given a vector $w \in \mathbb{R}^U$, we denote its marginal on the $i$th arm by $w_i \in \mathbb{R}^{U_i}$, defined as
$$w_i(u_i)=\sum_{v\in U: v_i=u_i} w(v).$$
In particular, for the probability vector $p$, the marginals are $p_i \in \Delta(U_i)$, namely, $p_i(u_i)=\mathbb{P}[X_i=u_i]$.

\paragraph{Optimal decision making.} 
If the agent takes action $i \in [n]$ (i.e., pulls arm $i$), her (unnormalized) expected utility is
$\sum_{u_i\in U_i} u_i w_i(u_i)$.
We denote by
$$u^*(w)=\max_{i\in N} \sum_{u_i\in U_i} u_i w_i(u_i)$$
the utility obtained from pulling the best arm, and by
$$N^*(w)=\{i\in N: \sum_{u_i\in U_i} u_i w_i(u_i)=u^*(w)\}$$
the set of best arms.

\paragraph{\textbf{The repeated environment.}} The actual environment is repeated: we have a sequence of information structures $I[t]=(p[t],\bq[t],\br[t])$ that describe information over (possibly changing) sets of realization profiles $U[t]$. The case in which $I[t]=I$ for every $t \in [T]$, for some fixed information structure $I$, is called a \emph{fixed environment}. Otherwise, we refer to the sequence as a \emph{changing environment}.

\begin{remark}
    All our positive results are stated for the general case of a changing environment, while all our negative results hold even in the special case of a fixed environment.
\end{remark}

\paragraph{\textbf{Strict Informational Superiority.}}
We begin by defining informational superiority in a single period. The more informed agent can always perform at least as well as the less informed one (for instance, by ignoring her additional information). However, in some cases, her performance may be no better, for example, when the extra information never changes the chosen action. In such cases, within the context of the decision problem, neither agent is truly superior. We therefore wish to distinguish between the two agents only when the performance of the more informed agent is strictly better. Formally, we have the following definition.

\begin{definition}[Strict informational superiority (static version)]\label{def:strict-static}
An information structure $I=(p,\bq,\br)$ satisfies \emph{strict informational superiority} if
\begin{align}\label{eq:strict}
    \sum_{j\in [l]} u^*(q^j) < \sum_{j,j'\in [l]} u^*(r^{j,j'}).
\end{align} 
In this case, we say that the agent $AG_1$ holding $\br$ is \emph{strictly more informed} than the agent $AG_2$ holding $\bq$, and we write $AG_1 \succ AG_2$.
\end{definition}

Note that the left-hand side of Equation \eqref{eq:strict} represents the expected utility of the less informed agent when acting optimally, while the right-hand side represents that of the more informed agent.

We now turn to the dynamic version, where the environment may change over time and is described by a sequence of information structures $I[t]$. In this setting, the notion of being \emph{strictly} more informed requires careful treatment. We adopt the following definition.

\begin{definition}[Strict informational superiority (dynamic version)]\label{def:dynamic-strict}
A sequence of information structures $I[t]=(p[t],\bq[t],\br[t])$ satisfies \emph{strict informational superiority} if there exists $\epsilon>0$ such that, for every $t\in [T]$,
\begin{align}\label{eq:epsilon-strict}
    \sum_{j\in [l]} u^*(q^j)+\epsilon < \sum_{j,j'\in [l]} u^*(r^{j,j'}).
\end{align} 
In this case, we say that the agent $AG_1$ holding $\br[t]$ is \emph{strictly more informed} than the agent $AG_2$ holding $\bq[t]$, and we write $AG_1 \succ AG_2$.\footnote{With a mild abuse of notation, $AG_1 \succ AG_2$ denotes strict informational superiority in both the static case (Definition~\ref{def:strict-static}) and the dynamic case (Definition~\ref{def:dynamic-strict}).}
\end{definition}

Namely, we require a universal $\epsilon$ such that $\epsilon$-superiority holds for all periods $t \in [T]$.\footnote{An alternative natural definition of strict informational superiority would simply require strict inequality, as in Equation~\eqref{eq:strict}, in every period. However, in the context of distinguishing between the less and more informed agents in the long run, such a definition is problematic. To see this, consider the case where, at time $t$, the more informed agent receives valuable information with a small and vanishing probability $\delta^t$; otherwise, both agents receive the same information. The probability that even one such rare event occurs is at most $\delta + \delta^2 + \ldots = \tfrac{\delta}{1-\delta}$, which is small. Nevertheless, the strict inequality in Equation~\eqref{eq:strict} would still hold for all $t \in \mathbb{N}$.} Notice that a repeated fixed environment, $I[t] = I$, satisfies informational superiority according to Definition~\ref{def:dynamic-strict} if and only if the information structure $I$ satisfies the (static) informational superiority condition of Definition~\ref{def:strict-static}.

\paragraph{\textbf{Decision-Maker and Adviser.}} Our model features two agents: a \emph{decision-maker} (DM) and an \emph{adviser} (AD). One of them is the more informed agent (holding the information $\br$), and the other is the less informed agent (holding the information $\bq$). As outside observers, we do not know which agent is more informed.

In every period $t \in [T]$, both agents choose an arm $a_{DM}, a_{AD} \in [n]$ based on their private information $\bq$ or $\br$. The arm $a_{DM}$ is actually pulled, and the realization of the random variable $X_{a_{DM}}$ is then observed. We denote this realization by $u_i \in U_i$ for $i=a_{DM}$, noting that $u_i$ is distributed according to the probability vector $\tfrac{w}{||w||1} \in \Delta(U_i)$, where $w = r^{j,j'}_{i}$ and $(j,j')$ are the signals observed by the more informed agent. This follows from the fact that $\br$ contains all the information in the system.

It is important to note that if $a_{AD} \neq a_{DM}$, the arm $a_{AD}$ is not pulled and its realization is not observed; the advice $a_{AD}$ thus remains an unexplored recommendation whose counterfactual consequences will never be observed.

We distinguish between the cases where the decisions $a_{DM}, a_{AD} \in [n]$ are made sequentially (\emph{sequential decisions}) or simultaneously (\emph{simultaneous decisions}). The incentives of DM and AD will be discussed later.

\paragraph{\textbf{Scoring Tests.}} A test can be interpreted as a procedure through which the public assesses the performance of DM and AD. In our model, the public observes only the actions and the chosen by DM arm realization in each period.

We denote by $O = [n] \times [n] \times \oU$ the set of possible action–realization triples, each referred to as an \emph{outcome}. In period~$t$, the test observes $o[t] = (a_{DM}[t], a_{AD}[t], u_i[t]) \in O$, where $i = a_{DM}$. The test then attempts to identify which agent is more informed based on the sequence of outcomes $(o[t])_{t \in \mathbb{N}} \in O^{\N}$.

We primarily restrict attention to a natural class of \emph{scoring tests}. A discussion of more general classes of tests is provided in Section~\ref{sec:generalTests}, where we also present a first step toward understanding the limitations of arbitrary tests (beyond scoring tests); see Proposition~\ref{pro:noTest}.

A \emph{scoring test} assigns a real number to every observation $o \in O$ and attempts to identify the winner based on the limit average of the scores. Formally, a \emph{scoring test $s$} is specified by a \emph{scoring function} $s:O \to \mathbb{R}$ together with a designated \emph{tie-breaking agent} (DM or AD). Let $L=\liminf_{T \to \infty} \frac{1}{T}\sum_{t\in [T]} s(o[t])$ be the limit score.\footnote{The choice of $\liminf$ (rather than, for example, $\limsup$ or any other rule for selecting among partial limits) plays no substantive role. It is merely a formality to ensure the test is well defined for all sequences of observations.}

\begin{itemize}
\item If $L>0$, the scoring test $s$ selects DM as the winner.
\item If $L<0$, the scoring test $s$ selects AD as the winner.
\item If $L=0$, the tie-breaking agent is selected.
\end{itemize}

Note that choosing 0 as the test’s threshold is without loss of generality. Indeed, given a test with a different threshold $\tau$, one can define $s'(o)=s(o)-\tau$, which has threshold 0 and selects the same winner.

Scoring tests are equivalent to the natural approach of assigning scores $s_{DM}(o),s_{AD}(o)$ for both agents based on the outcome, and picking the agent that has collected a higher sum of scores; indeed, one can set $s(o)=s_{DM}(o)-s_{AD}(o)$.

In Appendix~\ref{ap:axiom}, we give an axiomatic characterization of scoring tests. We show that tests satisfying \emph{empirical measurability} and \emph{convexity} characterize this class—up to the choice of a tie-breaking agent\footnote{One may additionally require that the winning set of an agent be either open or closed to establish a full characterization; see Observation~\ref{lem:sc2} in Appendix~\ref{ap:axiom}.} (see Observation~\ref{lem:sc} in Appendix~\ref{ap:axiom}). Empirical measurability requires that the test depend only on the empirical frequencies of the observed outcomes, thereby capturing the idea that all periods are treated equally. Convexity can be motivated by the following (hypothetical) scenario: suppose one tester observes the outcomes in odd periods only, while another observes those in even periods only. It is reasonable to require that if both testers unanimously select the same agent as the winner, then the test observing all periods should also select that agent as the winner.

\paragraph{\textbf{Agents' incentives and tests' induced welfare.}} We adopt an extreme perspective and assume that the \emph{sole} objective of both agents is to be selected by the test. In particular, neither the DM nor the AD is concerned with maximizing the welfare $u_i[t]$; they regard these realizations only insofar as they affect the result of the test. In other words, each agent is concerned only with securing a favorable track record by the end of the interaction, when their performances are compared.

Formally, every scoring test $s$ defines an infinitely repeated zero-sum game with incomplete information between the DM and the AD. After these periods, a winner is determined based on the limit score. We denote this zero-sum game by $\Gamma(s)$.

The value of $\Gamma(s)$ is denoted by $\val[\Gamma(s)] \in [0,1]$, and it represents the probability that the DM is selected by the test in an equilibrium of the zero-sum game. The welfare induced by the test (for the public) is denoted by $W(s)$ and is defined as the worst-case expected welfare across all equilibria of $\Gamma(s)$; namely,
\begin{align}\label{eq:w}
  W(s)=\min \Biggl\{ \liminf_{T\to \infty} \frac{1}{T} \sum_{t\in [T]} \E_{\eta,I[t]} [u_{a_{DM}[t]}]: \eta \text{ is and equilibrium of } \Gamma(s) \Biggl\}  
\end{align}
where the expectation is taken over the (possibly mixed) actions of the DM in equilibrium as well as over the random signals and welfare realizations in the model.\footnote{The choice of focusing on the worst-case equilibrium avoids situations in which the DM has no effect on the test and could therefore simply choose the best actions for the public.}

\paragraph{\textbf{Identification of the more-informed agent.}}
We are interested in whether the more-informed agent can be identified in the long run. This notion is defined as follows.

\begin{definition}\label{def:id-more}
    We say that a scoring test $s$ \emph{identifies the more-informed agent} if, for every sequence of information structures $(I[t])_{t\in \N}$ for which $DM \succ AD$, it holds that $\val[\Gamma(s)]=1$, and for every sequence of information structures $(I[t])_{t\in \N}$ for which $DM \prec AD$, it holds that $\val[\Gamma(s)]=0$.
\end{definition}

In other words, the more informed agent should be able to prove it to the public (i.e., to guarantee being selected) in the long run.

\section{Results}\label{sec:results}

We now state our main results. Intuition is provided in Section~\ref{sec:intuition}, while the formal proofs are presented in Appendix~\ref{ap:proofs}.
Except for Section~\ref{sec:seq}, which considers sequential decisions, we focus on the simultaneous-decision setting; that is, in every period $t \in \N$, the actions $a_{DM}[t]$ and $a_{AD}[t]$ are chosen simultaneously. We begin with the following positive result.

\begin{theorem}\label{theo:simul}
    There exists a scoring test that identifies the more-informed agent in the setting with simultaneous decisions and a changing environment.
\end{theorem}

Theorem~\ref{theo:simul} establishes that even when the information structure changes over time, one can design a simple \emph{scoring} test that identifies the more-informed agent. The result holds for any number of arms $n$.

We now describe the scoring test proposed in Theorem~\ref{theo:simul} to identify the more-informed agent. The tie-breaking agent (i.e., the one selected in case of a limit score of 0) is the DM. The scoring function is defined as
\begin{align}\label{eq:scor}
    s(a_{DM},a_{AD},u_i)=\begin{cases}
        u_i &\text{ if } a_{DM}=a_{AD}, \\
        -\frac{u_i}{n-1} &\text{ otherwise,}
    \end{cases}
\end{align}
so that the DM aims to match the action of the AD, while the AD aims to mismatch. The factor $\tfrac{1}{n-1}$ compensates the DM for the fact that matching becomes more difficult as the number of arms increases ($n \geq 2$). A key feature is that the stakes of this matching game depend on the realization, which is directly influenced by the DM but not by the AD.

The intuition for this choice of scoring test is discussed in Section~\ref{sec:intuition}, while the formal proof of Theorem~\ref{theo:simul} is presented in Appendix~\ref{sec:sim-proof}.

\subsection{A Finite-Horizon Analogue}

The infinite-horizon scoring test of Equation~\eqref{eq:scor} admits a finite-horizon analogue. Concretely, the test selects a winner after $T$ periods and is defined as follows.\footnote{The exact choice of threshold $-T^{3/4}$ is not crucial; any threshold of the form $-T^{\alpha}$ with $1/2<\alpha<1$ works as well.}
\begin{align}\label{eq:finite}
    \text{Selects the DM as the winner } \Leftrightarrow  \sum_{t\in T} s[o_t]\geq -T^{3/4}.
\end{align}

\begin{corollary}\label{cor:finite}
    The finite-horizon scoring test in \eqref{eq:finite} selects the more-informed agent after $T$ periods with probability at least $1-\exp(-\sqrt{T})$ in every changing environment. 
\end{corollary}

Corollary \ref{cor:finite} is simply an application of Theorem \ref{theo:simul} combined with a concentration inequality. The formal proof is relegated to Appendix \ref{ap:proof-cor}.

\subsection{Interplay Between Identification and Welfare}\label{sec:welfare}

Beyond the objective of identifying the more-informed agent, there is another, no less important, goal: ensuring that the decisions made by the DM are reasonably good for the public. The scoring function of Equation~\eqref{eq:scor} performs poorly in this respect, as illustrated by the following example.

\begin{example}\label{ex:test-perf}
    Consider a binary-arm environment ($n = 2$) in which it is common knowledge that the realizations are deterministic and given by the profile $u = (K, 1)$ for $K \gg 1$. The welfare-maximizing choice for the DM is to take action~1, yielding a welfare level of~$K$. However, under the incentives induced by the test, the agents end up playing the following zero-sum game.
    
\begin{table}[ht]
\begin{center}
\begin{tabular}{ccc}
                       & 1                        & 2                         \\ \cline{2-3} 
\multicolumn{1}{c|}{1} & \multicolumn{1}{c|}{$K$} & \multicolumn{1}{c|}{$-K$} \\ \cline{2-3} 
\multicolumn{1}{c|}{2} & \multicolumn{1}{c|}{-1}  & \multicolumn{1}{c|}{1}    \\ \cline{2-3} 
\end{tabular}
 \end{center}
\end{table}

In equilibrium, the DM takes the welfare-superior action 1 with probability only $\tfrac{1}{K+1}$, which vanishes as $K$ grows, implying severely suboptimal welfare for large $K$.
\end{example}

This example is not a coincidence. We show that if a scoring test succeeds in identifying the more-informed agent, then there necessarily exists an instance in which the DM’s equilibrium decisions (induced by the test) perform poorly with respect to public welfare. In other words, the objective of identifying the more-informed agent conflicts with the objective of incentivizing high-welfare decisions, and the two cannot be reconciled.

Let $U^*_{DM}$ denote the welfare attainable by the DM if she acted to maximize welfare rather than seeking to be selected by the test. Formally, if the DM is the less-informed agent, then  
\begin{align}\label{eq:udm1}
 U^*_{DM}=\liminf_{T\to \infty} \frac{1}{T}\sum_{t\in [T]} \sum_{j\in [l]} u^*(q^j[t]),   
\end{align}
and if the DM is the more-informed agent, then
\begin{align}\label{eq:udm2}
U^*_{DM}=\liminf_{T\to \infty} \frac{1}{T}\sum_{t\in [T]} \sum_{j,j'\in [l]} u^*(r^{j,j'}[t]).    
\end{align}
Recall that the welfare for the public (in equilibrium) induced by the scoring test $s$ is denoted by $W(s)$ (see Equation~\eqref{eq:w}). We say that a scoring test $s$ \emph{achieves a $c$-fraction of the welfare} if $W(s)\geq c \cdot U^*_{DM}$ for every fixed sequence of information structures $I[t]=I$.

\begin{theorem}\label{theo:welfare}
  For every $c > \frac{1}{2}$, in the setting with simultaneous decisions, binary arms ($n=2$), and a fixed environment, there exists no scoring test that both identifies the more-informed agent and achieves a $c$-fraction of the welfare.
\end{theorem}

Observe that for binary arms, a $1/2$-fraction of $U^*_{DM}$ is a very low benchmark: it can be attained by choosing an arm uniformly at random without relying on any information. The proof intuition is discussed in Section~\ref{sec:intuition}, while the formal proof of Theorem~\ref{theo:welfare} is presented in Appendix~\ref{sec:welfare-proof}.

Given the negative result of Theorem \ref{theo:welfare}, it is natural to ask whether one can compromise on the probability of selecting the more informed agent and increase the welfare performance. For scoring tests, this question is meaningless because the selection is done with probability 0 or 1; see the reduction to a one-shot interaction (Lemmas \ref{lem1} and \ref{lem2}). However, for tests beyond the scoring ones, such a compromise is very effective; see Proposition \ref{pro:wel}.

\subsection{Sequential Decisions}\label{sec:seq}

In this section, we consider the case where, in every period $t \in \N$, the \emph{leader}, an agent $LE \in \{DM,AD\}$, first chooses an action $a_{LE} \in [n]$. After observing $a_{LE}$, the \emph{follower} agent $FO \in \{DM,AD\}, FO \neq LE$ chooses an action $a_{FO} \in [n]$. Finally, the welfare is realized. This order is fixed across all periods $t \in \N$. Both variants (DM acting first or AD acting first) are referred to as the \emph{sequential-decision setting}.

For the sequential-decision setting, we obtain an impossibility result, regardless of any additional requirements imposed on the test (such as the welfare condition in Theorem~\ref{theo:welfare}).

\begin{theorem}\label{theo:seq}
    In the setting with sequential decisions and a fixed environment, no scoring test exists that identifies the more-informed agent. 
\end{theorem}
This impossibility result holds already for binary arms ($n=2$). The proof intuition is discussed in Section~\ref{sec:intuition}, while the formal proof of Theorem~\ref{theo:seq} is presented in Appendix~\ref{sec:seq-proof}.

\subsubsection{Single Safe Arm and Single Risky Arm}

The case of a single safe arm ($|U_1|=1$) and a single risky arm ($|U_2|>1$) is of special interest in our settings because it captures a natural scenario in which risky projects arrive sequentially, and the DM simply decides whether to accept or reject each of them. The impossibility result of Theorem \ref{theo:welfare} on the interplay with welfare is valid in this basic environment. The impossibility result of Theorem \ref{theo:seq} for sequential decisions is not valid. In fact, there is a simple scoring test that identifies the more informed agent. 

\begin{proposition}\label{pro:ad1}
    In a setting with binary arms, a single safe arm, a single risky arm, sequential decisions where the AD acts first, and a changing environment, there exists a scoring test that identifies the more-informed agent.
\end{proposition}

The scoring test is quite intuitive. If DM chooses the safe arm $a_{DM}=1$, the resulting score is 0; i.e., $s(1,\cdot,c)=0$ when $c$ is the outcome in the safe arm. If DM chooses the risky arm $a_{DM}=2$, the score is according to AD's regret. Concretely, we set $s(2,2,u_2)=c-u_2$ and $s(2,1,u_2)=u_2-c$. Namely, if the advice was to take the risky arm, the AD's regret is $c-u_2$. If the advice was to take the safe arm, his regret is $c-u_2$. The tie-breaking agent is the AD.

Intuitively, if the AD is more informed, he can choose the welfare-maximizing arm in each period, which guarantees non-positive expected regret. The DM has no extra information, so she has no course of action to make this regret positive. But the AD is the tie-breaking agent, so he is chosen. If, on the other hand, the DM is more informed, she can take the risky action in all periods in which she identifies, with her extra information, that AD's action is welfare inferior. In all other periods, the DM takes the safe action. This induces a positive regret on the AD. The formal proof is relegated to Appendix~\ref{ap:proofs}.

Interestingly, if we consider a setting in which the DM acts first, we get back to an impossibility result.

\begin{proposition}\label{pro:dm1}
    In a setting with binary arms, a single safe arm, a single risky arm, sequential decisions where the DM acts first, and a fixed environment, there exists no scoring test that identifies the more-informed agent.
\end{proposition}

To gain some intuition, let us explain why the idea behind the positive result of Proposition \ref{pro:dm1} fails. The AD might know that the DM's safe action is inferior. However, he cannot prove it to the public because the public will not observe the outcome of the risky arm. The formal proof is relegated to Appendix \ref{ap:proofs}.

\section{Intuition Behind the Proofs}\label{sec:intuition}

The results in Section~\ref{sec:results}, both positive and negative, are founded on several common key ideas, which are presented below.

Given the natural structure of scoring tests, where scores from different periods are aggregated additively, the repeated interaction might be expected to be closely related, and in fact equivalent, to the one-shot interaction. We begin by formally defining the one-shot interaction.

Let $s:O \to \mathbb{R}$ be a scoring function and $I=(p,\bq,\br)$ an information structure, where the roles of DM and AD are assigned to $\bq$ (less informed) and $\br$ (more informed). Define $G(s,I)$ as a one-shot zero-sum game with incomplete information, in which the agents choose actions $a_{DM}$ and $a_{AD}$ (simultaneously or sequentially, depending on the context), and the payoff is determined by $s$. We refer to this one-shot analogue of identification (Definition~\ref{def:id-more}), which applies to repeated interaction, as \emph{separation}.

\begin{definition}\label{def:separation}
    A scoring function $s$ with DM (respectively, AD) serving as the tie-breaking agent is called \emph{separating} if, for every information structure $I$, the following conditions hold:
\begin{itemize}
\item If $DM \succ AD$, then $\val[G(s,I)] \ge 0$ (respectively, $\val[G(s,I)] > 0$), and
\item If $DM \prec AD$, then $\val[G(s,I)] < 0$ (respectively, $\val[G(s,I)] \le 0$).
\end{itemize}
A scoring function $s$ is called \emph{separating} if there exists a choice of tie-breaking agent for which it satisfies the above conditions. 
\end{definition}

The following two lemmas establish the connection between the one-shot interaction (separation) and the repeated interaction (identification).

\begin{lemma}\label{lem1}
If a scoring function $s$ is separating, then the corresponding scoring test $s$ (with the same tie-breaking agent) identifies the more-informed agent in every changing environment.
\end{lemma}

\begin{lemma}\label{lem2}
If no separating scoring function exists, then no scoring test can identify the more-informed agent, even in a fixed environment.
\end{lemma}

Lemmas~\ref{lem1} and~\ref{lem2} together establish the equivalence between the existence of a separating scoring function and the existence of an identifying scoring test. Moreover, when such a function exists, the result holds in the general setting of a changing environment, whereas its nonexistence implies a negative result even in the special case of a fixed environment. This distinction clarifies the scope of the positive and negative results: the positive result (Theorem~\ref{theo:simul}) holds for changing environments, while the negative results (Theorems~\ref{theo:welfare} and~\ref{theo:seq}) remain valid even under a fixed environment.

The arguments underlying Lemmas~\ref{lem1} and~\ref{lem2} are relatively straightforward. For Lemma~\ref{lem1}, if the DM can guarantee a positive value in each round, she can also guarantee that the limit average remains positive, and the same reasoning applies to the AD. For Lemma~\ref{lem2}, the absence of a separating scoring function implies that, for every scoring function, there exists a problematic instance $I$. The repeated environment with $I[t]=I$ then constitutes a counterexample to identification. The only subtle point concerns the distinction between weak and strict inequalities. In particular, in Lemma~\ref{lem1}, for a changing environment, we must carefully ensure that the limit average is strictly positive (or strictly negative). The formal proofs of these lemmas are provided in Appendix~\ref{ap:missing}.

Lemmas~\ref{lem1} and~\ref{lem2} allow us to focus on one-shot interactions (separation) to derive our results. Our first goal is to provide the necessary conditions for a function to be separating. To this end, we consider a \emph{complete-information} game $G(s,p)$, defined as follows. In the game $G(s,p)$, both agents share a common prior $p$ and receive no additional information; that is, $G(s,p)$ corresponds to the game $G(s,I)$ with $I=(p,\bq,\br)=(p,p,p)$. We begin with the following necessary condition. 

\begin{lemma}\label{lem:val0}
    If $s$ is separating, then $\val [G(s,p)]=0$ for every prior $p$. 
\end{lemma}

Before proving Lemma~\ref{lem:val0}, we define the following sets for later reference. For each arm $i$, let $A_i \subset \Delta(U)$ denote the set of priors under which arm $i$ is optimal, namely, $A_i=\{p \in \Delta(U): i\in N^*(p)\}$. We further denote by $\inter(A_i) = A_i \setminus \bigcup_{i' \neq i} A_{i'}$ the interior of $A_i$, that is, the set of priors for which arm $i$ is the \emph{unique} optimal arm.

\begin{proof}[Proof of Lemma~\ref{lem:val0}]
Observe that $\bigcup_i \inter(A_i)$ is dense in $\Delta(U)$, and that $\val[G(s,p)]$ is a continuous function of $p$. Therefore, it suffices to prove the lemma for $p \in \bigcup_i \inter(A_i)$. Without loss of generality, assume that $p \in \inter(A_1)$.

Assume, by way of contradiction, that $\val[G(s,p)] = c \neq 0$. We consider the case $c > 0$; the argument for $c < 0$ is analogous. Since the scoring function $s$ is defined over a finite set $O$, it is bounded; that is, $|s(o)| \le M$ for every $o \in O$.

Consider an information structure $I$ in which the AD has posterior $p$ with probability $1-\epsilon$ and a different posterior $p' \in \inter(A_2)$ with probability $\epsilon$. In other words, with probability $\epsilon$, the AD receives information that makes action $2$ optimal instead of action $1$. The DM receives no information, so her prior (and hence her posterior) is $(1-\epsilon)p + \epsilon p'$. Observe that in this information structure $I$, the AD is strictly more informed. However, the DM can employ her optimal strategy from the game $G[s,p]$ and guarantee a score of at least $(1-\epsilon)c - \epsilon M$, which is strictly positive for sufficiently small $\epsilon$. This contradicts the assumption that $s$ is separating, since under $I$ the DM, who is strictly less informed, obtains a strictly positive score.  
\end{proof}

The next step towards a characterization of separating scoring functions is to examine the optimal strategies of the agents in the games $G(s,p)$. More specifically, we focus on the optimal strategies of the tie-breaking agent. For clarity of exposition, we assume that the tie-breaking agent is the DM; the arguments remain unchanged if the tie-breaking agent were the AD. To convey the intuition behind our proofs, we introduce the following simplifying assumption. In the formal proofs, provided in Appendix~\ref{ap:proofs}, this assumption is not required.

\begin{assumption}\label{as:unique}
The scoring function $s$ is such that the DM has a unique optimal strategy in the game $G(s,p)$ for every $p \in \Delta(U)$.
\end{assumption}

Under Assumption~\ref{as:unique}, we can define a function $f:\Delta(U) \to \Delta([n])$ that maps each prior belief to the DM’s optimal strategy in the game $G(s,p)$. For clarity of exposition, we introduce another simplifying assumption that, again, is not required in the formal proofs.

\begin{assumption}\label{as:posteriors}
The information structure $I$ is such that all posteriors of the more-informed agent belong to $\bigcup_{i \in [n]} \inter(A_i)$.
\end{assumption}

In other words, the more-informed agent always has a unique welfare-maximizing action. Under Assumptions~\ref{as:unique} and~\ref{as:posteriors}, we can state a necessary and sufficient condition for separation that depends only on the class of complete-information games ${G(s,p) : p \in \Delta(U)}$, rather than on the more complex class of incomplete-information games ${G(s,I)}$. 

\begin{lemma}\label{lem:equiv}
Under Assumptions~\ref{as:unique} and~\ref{as:posteriors}, a scoring function $s$ is separating if and only if $\val[G(s,p)] = 0$ for every $p \in \Delta(U)$, and $f(p) \neq f(p')$ for every $p \in \inter(A_i)$ and $p' \in \inter(A_{i'})$ with $i \neq i'$.
\end{lemma}

That is, separation requires a zero value in every interaction with no information and distinct optimal strategies across welfare-maximizing regions.

\begin{proof}[Proof of Lemma~\ref{lem:equiv}]
We first show that if the stated condition does not hold, then $s$ is not separating. Without loss of generality, assume that $\{i, i'\} = \{1, 2\}$; that is, there exist $p \in \inter(A_1)$ and $p' \in \inter(A_2)$ such that $f(p) = f(p')$. Consider the information structure $I$ in which the AD’s posterior is $p$ with probability $1/2$ and $p'$ with probability $1/2$. The DM is uninformed, and her prior (and thus her posterior) is $\frac{1}{2}p + \frac{1}{2}p'$. Observe that in $I$, the DM is strictly less informed. Nevertheless, the DM can use the strategy $f(p)$ and guarantee an expected score of $0$, regardless of the AD’s private information. Since the DM is the tie-breaking agent, guaranteeing a score of $0$ suffices to violate separation.

    For the opposite direction, assume that the condition holds. We now show that $s$ is separating. If the DM is the more-informed agent, she observes the AD’s private information $q^j$. The DM can then ignore her own private information and simply play the strategy $f(s, q^j)$, which guarantees her an expected score of $0$, since $\val[G(s, q^j)] = 0$ for every $q^j$.

    If the AD is strictly more informed, we now show that he can guarantee a strictly negative expected score. By the minimax theorem, it suffices to show that for every strategy of the DM, the AD can best respond to it and obtain a strictly negative score. Since the AD is strictly more informed, with positive probability a signal $q^j$ of the DM is refined into posteriors $r^{j,j'} \in \inter(A_{i'})$ and $r^{j,j''} \in \inter(A_{i''})$ with $i' \neq i''$. By the stated condition, $f(r^{j,j'}) \neq f(r^{j,j''})$. In particular, based on her information $q^j$, the DM cannot simultaneously play $f(r^{j,j'})$ when the AD’s posterior is $r^{j,j'}$ and $f(r^{j,j''})$ when it is $r^{j,j''}$.

Assume that when the AD’s posterior is $r^{j,j'}$, the DM does not play $f(r^{j,j'})$, that is, she does not play her optimal strategy in the game $G(s, r^{j,j'})$. In this case, the AD can exploit this deviation and induce a strictly negative expected score for the DM at $r^{j,j'}$. For all remaining posteriors ${r^{k,k'} : (k,k') \neq (j,j')}$, the AD can use his optimal strategy in the corresponding game $G(s, q^k)$ to guarantee a 0 expected score. Hence, in expectation, the DM’s score is strictly negative.
\end{proof}

We now discuss how the observations derived above are used to establish the main results.

\paragraph{\textbf{Proof Idea of Theorem~\ref{theo:simul}.}}
According to Lemma~\ref{lem:equiv}, we need a scoring function $s$ for which the DM’s optimal strategy in the game $G(s,p)$ reveals the welfare-maximizing arm. For a given prior $p$, let $e_i = \mathbb{E}_p[X_i]$ denote the expected welfare of arm $i$. For the scoring function $s$ defined in Equation~\eqref{eq:scor} and a given prior $p$, the corresponding zero-sum game $G(s,p)$ is displayed in Table~\ref{tab:g}.

\begin{table}[ht]
\begin{center}
\begin{tabular}{cccccc}
                              & 1                                       & 2                                       & 3                                       & $\cdots$                      & $n$                                     \\ \cline{2-6} 
\multicolumn{1}{c|}{1}        & \multicolumn{1}{c|}{$e_1$}              & \multicolumn{1}{c|}{$-\frac{e_1}{n-1}$} & \multicolumn{1}{c|}{$-\frac{e_1}{n-1}$} & \multicolumn{1}{c|}{$\cdots$} & \multicolumn{1}{c|}{$-\frac{e_1}{n-1}$} \\ \cline{2-6} 
\multicolumn{1}{c|}{2}        & \multicolumn{1}{c|}{$-\frac{e_2}{n-1}$} & \multicolumn{1}{c|}{$e_2$}              & \multicolumn{1}{c|}{$-\frac{e_2}{n-1}$} & \multicolumn{1}{c|}{$\cdots$} & \multicolumn{1}{c|}{$-\frac{e_2}{n-1}$} \\ \cline{2-6} 
\multicolumn{1}{c|}{3}        & \multicolumn{1}{c|}{$-\frac{e_3}{n-1}$} & \multicolumn{1}{c|}{$-\frac{e_3}{n-1}$} & \multicolumn{1}{c|}{$e_3$}              & \multicolumn{1}{c|}{$\cdots$} & \multicolumn{1}{c|}{$-\frac{e_3}{n-1}$} \\ \cline{2-6} 
\multicolumn{1}{c|}{$\vdots$} & \multicolumn{1}{c|}{$\vdots$}           & \multicolumn{1}{c|}{$\vdots$}           & \multicolumn{1}{c|}{$\vdots$}           & \multicolumn{1}{c|}{$\ddots$} & \multicolumn{1}{c|}{$\vdots$}           \\ \cline{2-6} 
\multicolumn{1}{c|}{$n$}      & \multicolumn{1}{c|}{$-\frac{e_n}{n-1}$} & \multicolumn{1}{c|}{$-\frac{e_n}{n-1}$} & \multicolumn{1}{c|}{$-\frac{e_n}{n-1}$} & \multicolumn{1}{c|}{$\cdots$} & \multicolumn{1}{c|}{$e_n$}              \\ \cline{2-6} 
\end{tabular}\caption{The game $G(s,p)$ when $e_i=\E_p[X_i]$.}\label{tab:g}
\end{center}
\end{table}

\noindent
One can verify that the \emph{unique} optimal strategy for the DM (the maximizing row player) is $(\tfrac{c}{e_1}, \ldots, \tfrac{c}{e_n})$, where $c = \big(\tfrac{1}{e_1} + \cdots + \tfrac{1}{e_n}\big)^{-1}$; this result is formally established in Lemma~\ref{lem:opt}. Observe that this optimal strategy reveals the welfare-maximizing arm (i.e., the arm with the highest $e_i$), since the best arm corresponds to the action assigned minimal weight in the optimal strategy.

\paragraph{\textbf{Proof Idea of Theorem~\ref{theo:welfare}.}}
We consider a simple setting with one safe action yielding a deterministic welfare and one risky action with two possible realizations. Lemma~\ref{lem:equiv} implies that the optimal strategy of the tie-breaking agent must differ between the regions $A_1$ (where the safe action is optimal) and $A_2$ (where the risky action is optimal). We show that such a change in the behavior of the tie-breaking agent can occur only if all games $G(s,p)$ admit a mixed-strategy equilibrium. Intuitively, if this is not the case, one of the agents must have a dominant strategy in the $2\times2$ game, which persists for nearby priors, and this persistence consequently leads to a contradiction. 

By normalizing the positive score of the safe action to $1$, the game $G(s,p)$ necessarily takes the following form:

\begin{table}[ht]
\begin{center}
\begin{tabular}{ccc}
                       & 1                              & 2                             \\ \cline{2-3} 
\multicolumn{1}{c|}{1} & \multicolumn{1}{c|}{1}         & \multicolumn{1}{c|}{$-c$}       \\ \cline{2-3} 
\multicolumn{1}{c|}{3} & \multicolumn{1}{c|}{$-L_1(p)$} & \multicolumn{1}{c|}{$L_2(p)$} \\ \cline{2-3} 
\end{tabular}   
\end{center}
\end{table} 

Here, $c>0$ is a constant, and $L_1, L_2 > 0$ are functions that are linear in $p$. By the indifference principle, the DM’s optimal strategy selects the risky arm with probability $\frac{1}{1 + L_1(p)}$. We show that, in a worst-case instance, such behavior may approximate the welfare-optimal (deterministic) decision by as bad factor as $\frac{1}{2}$.\footnote{A $(\frac{1}{2} - \epsilon)$ approximation is achievable by setting $L_1(p) \in [1 - \epsilon, 1 + \epsilon]$; moreover, such a choice preserves the property that different posteriors correspond to different optimal strategies for the DM, and hence the corresponding scoring function is separating.}

\paragraph{\textbf{Proof Idea of Theorem~\ref{theo:seq}.}}
A central object in the proof is the (hypothetical) zero-sum game $H_i$ played between the follower and \emph{nature}, in which nature chooses the realizations of the arms adversarially to minimize the follower’s score. The game $H_i$ becomes relevant after the leader has chosen action $i$.\footnote{This aspect of the proof is analogous to the classical literature on expert testing \citep{Sandroni2003Reproducible}, where the minimax argument of a game with an adversarial nature plays a central role.}

Roughly speaking, on the one hand, we cannot have $\val[H_i] > 0$, since in that case, for the prior corresponding to nature’s optimal strategy, the leader would be able to obtain a strictly positive score in the game $G(s,p)$. On the other hand, we cannot have $\val[H_i] < 0$, since then the follower could guarantee a strictly negative score after the leader’s choice of action $i$, even without possessing any information. This would render action $i$ dominated for the leader and, consequently, irrelevant. 

The fact that $\val[H_i] = 0$ for all actions $i \in [n]$ implies that the follower can guarantee an expected payoff of zero without relying on any information. This observation rules out the possibility that the follower acts as the tie-breaking agent, and hence the tie-breaking role must be assigned to the leader. To exclude this remaining case, we examine the sets of nature’s optimal strategies in the games $H_1, \ldots, H_n$. We establish that these sets collectively cover the entire simplex $\Delta(U)$. Furthermore, by arguments analogous to those used in Lemma~\ref{lem:equiv}, each such set must be contained within one of the sets $A_1, \ldots, A_n$. To complete the proof, we consider an instance with two non-deterministic arms and make use of the fact that the realizations of unchosen arms remain unobserved.

\section{Discussion}
\subsection{More General Classes of Tests}\label{sec:generalTests}

Although scoring tests form a natural class to study, particularly because they are the only ones satisfying \emph{empirical measurability} and \emph{convexity} (see the discussion in Section~\ref{sec:model} and Appendix~\ref{ap:axiom}), it is nevertheless desirable to develop a broader understanding of more general classes of tests, namely, to characterize the limits of identification of the more-informed agent beyond the scoring framework.

A general test is a mapping $\theta: O^{\mathbb{N}} \to \{DM, AD\}$. Discussing tests at this level of generality is problematic, as the induced game between the DM and the AD may fail to admit an equilibrium. For instance, without assuming that $\theta$ is Borel measurable, an equilibrium need not exist.\footnote{Our framework includes, as a special case, the setting of the Borel determinacy theorem. In such settings, an equilibrium may fail to exist when the relevant winning set of the DM is non-measurable; see, for example, \citep[page 137]{kechris2012classical}.} Therefore, imposing additional well-behaved restrictions on the test is essential.

To avoid issues related to equilibrium existence, which lie beyond the scope of this paper, we focus on \emph{well-behaved tests}, that is, tests $\theta$ for which the infinitely repeated game $\Gamma(\theta)$ admits an equilibrium for every sequence of information structures $(I[t])_{t \in \mathbb{N}}$.

We begin by noting that the positive result of Theorem~\ref{theo:simul} on the possibility of identification \emph{can} be reconciled with welfare maximization when considering general tests. This stands in contrast to the case of scoring tests, for which achieving both objectives—identification and welfare maximization—is impossible (Theorem~\ref{theo:welfare}).

\begin{proposition}\label{pro:wel}
    For every $\epsilon>0$, there exists a test that identifies the more informed agent with probability $1-\epsilon$, and achieves a 1-fraction of the welfare in a setting with simultaneous decisions and a changing environment. 
\end{proposition}

\begin{proof}
 For general tests, the objectives of identification and welfare maximization can be separated. The testing of which agent is more informed can be performed only in a sequence of periods with a vanishing frequency, for example, in the periods $\{t^2 : t \in \mathbb{N}\}$. Since this sequence is infinite, the test still succeeds in identifying the more-informed agent. In the remaining periods, $\mathbb{N} \setminus \{t^2 : t \in \mathbb{N}\}$, the DM can be incentivized to choose welfare-maximizing actions. This can be achieved, for instance, by committing that with probability $\epsilon > 0$, the test will select the DM with a probability that is strictly increasing in the welfare realized in each of these rounds. Identification then occurs with probability at least $(1 - \epsilon)$, while the resulting long-run welfare is asymptotically as high as that obtained under welfare-maximizing decisions.   
\end{proof}

We now establish a negative result (under relatively restrictive assumptions) that applies to all well-behaved tests. Although the result itself is limited, its proof introduces an imitation technique that may be useful in establishing other negative results for classes of tests beyond the scoring ones. We say that a test \emph{incentivizes optimal decisions} if there exists an equilibrium of the game $\Gamma(\theta)$ in which, in every period $t \in \mathbb{N}$, the DM takes a welfare-maximizing action $a_{DM}[t] \in N^*(w)$ given her belief $w$. We obtain the following negative result, which applies universally to all well-behaved tests.

\begin{proposition}\label{pro:noTest}
    In the setting with sequential decisions, where the DM acts first and the environment is fixed, there exists no well-behaved test that both identifies the more-informed agent and incentivizes optimal decisions.
\end{proposition}

Since the proof applies to all tests, the observations of Section~\ref{sec:intuition} are not relevant here. Instead, a simple imitation argument can be employed.

\begin{proof}[Proof of Proposition \ref{pro:noTest}]
    Let arm~1 be the safe arm, which yields a deterministic welfare $U_1 = \{2\}$, and let arm~2 be the risky arm, which yields one of two possible realizations $U_2 = \{1, 3\}$. The prior is uniform, $p = \tfrac{1}{2}$, representing the probability of the high realization~$3$ for the risky arm.

    Consider two information structures, $I$ and $I'$. In $I$, the AD is uninformed, while the DM receives a signal that correctly reveals the state with probability $\tfrac{3}{4}$. Hence, the DM’s posterior is $\tfrac{1}{4}$ or $\tfrac{3}{4}$, each occurring with probability $\tfrac{1}{2}$. In $I'$, the DM again receives a signal that is correct with probability $\tfrac{3}{4}$, but the AD also observes this signal. Moreover, in $I'$, whenever the DM’s posterior is $\tfrac{1}{4}$, the AD receives an additional signal that perfectly reveals the realization of arm~2. Notice that $DM \succ AD$ in information structure~$I$, whereas $DM \prec AD$ in information structure~$I'$.

    Since we have assumed that the test incentivizes optimal decisions, from the perspective of an outside observer (the test), in both information structures $I$ and $I'$, the following occurs: in every period, the DM selects the risky arm with probability $\tfrac{1}{2}$, and when it is selected, the high realization occurs with probability $\tfrac{3}{4}$.

    If the test identifies the more-informed agent, then the AD must have a strategy that ensures he reaches his winning set in information structure $I'$. However, the AD can apply the same strategy to reach his winning set also in information structure $I$, thereby contradicting the assumption that the test identifies the more-informed agent DM in information structure $I$.
\end{proof}

It remains an interesting open problem to characterize the extent to which general classes of tests can or cannot identify the more-informed agent. A key challenge in establishing such results lies in characterizing the equilibrium strategies of the agents in the resulting infinitely repeated games. Below, we discuss two classes of tests that may serve as initial steps toward understanding this intricate open problem.

Let us focus on empirically measurable tests and assume that the DM’s winning set $W_{DM} \subset \Delta(O)$, is convex; i.e., the set of limit frequencies that select $DM$ is $W_{DM}$. However, we omit the assumption that the AD’s winning set, $W_{AD} = \Delta(O) \setminus W_{DM}$, is also convex (otherwise, these would be scoring tests as shown in Appendix \ref{ap:axiom}). This setting constitutes a special case of the Blackwell approachability theorem \citep{blackwell1956analog}. The optimal strategies of the DM in this framework are somewhat complex—for instance, they are generally non-stationary. Nevertheless, the theory provides a relatively good understanding of approaching strategies in the Blackwell approachability problem. This class of tests may therefore serve as a promising starting point for extending our results to broader families of empirically measurable tests.

Another interesting class of tests, which includes the empirically measurable tests as a subset, is that of \emph{measurable} and \emph{shift-invariant} tests. Shift invariance requires that the test’s outcome remain unchanged by a single-period shift; formally, $\theta(h) = \theta((o,h))$ for every history $h \in O^{\mathbb{N}}$, where $(o,h)$ denotes the history obtained by adding outcome $o$ as the first period and shifting all other outcomes one period forward. Shift invariance reflects the natural requirement that a single period (out of infinitely many) should not affect the test’s result. The existence of (approximate) equilibria in such zero-sum games has been established by \citep{flesch2023equilibrium}, in settings far more general than ours. Understanding the identification power of such tests remains an interesting open question.

\subsection{Crucial Aspects of the Negative Results}

The message conveyed by the main results is largely negative:

\begin{itemize}
    \item Even though, in the case of simultaneous decisions, the more-informed agent can be identified, doing so necessarily incentivizes the DM to act suboptimally (potentially performing no better than uniform random decision-making).
    \item In the case of sequential decisions, identifying the more-informed agent is simply impossible.
\end{itemize}

Below, we highlight several features of our model that drive these negative results. For each feature, we present an alternative version of the model in which modifying it makes identification of the more-informed agent feasible and aligned with welfare-maximizing behavior by the DM, thereby illustrating its crucial role in the negative findings.

\paragraph{\textbf{Unobserved counterfactuals.}} If counterfactuals were observable, identifying the more-informed agent would be trivial: one could simply select the agent whose decisions yield higher (limit-average) welfare. Moreover, such a test would automatically incentivize the DM to maximize welfare.

\paragraph{\textbf{Unknown information structure.}} If the information structure $I = (p, \bq, \br)$ were known to the test, identifying the more-informed agent would again be trivial. The test could simply verify whether the DM’s (limit-average) welfare matches the optimal welfare attainable under information $\br$. A DM possessing $\br$ would pass this test, whereas a DM possessing $\bq$ would not. Moreover, such a test would naturally align incentives toward welfare-maximizing behavior.

\paragraph{\textbf{Coarse actions.}} In our setting, the agents communicate with the test only through their choice of arms. One may instead consider an alternative framework in which the agents are also required to send a cheap-talk message to the test. In such a setting, a test that identifies the more-informed agent can be constructed.\footnote{The central role of the coarse actions is analogous to the findings in the social-learning literature \citep{banerjee1992simple, bikhchandani1992theory, bikhchandani2024information}, where the coarse-action assumption—namely, that agents observe only the actions of their peers rather than their underlying beliefs—leads to inefficiencies such as herding.} For example, the test could ask both agents to provide forecasts of the expected realization corresponding to the DM’s chosen arm. More specifically, the AD would provide $n$ such forecasts in the case where the DM’s action is unknown to him. The test would then select the agent whose (limit-average) squared distance between forecasts and realized welfare is smaller. Since the squared-distance rule is a \emph{strictly proper scoring rule} (see \citep{gneiting2007strictly}), the more-informed agent would win this competition.

To incorporate welfare-maximization incentives for the DM into this test, one can use the same argument as in Proposition \ref{pro:wel}. We assign a small probability $\epsilon > 0$ for this purpose and design the probability of the DM’s selection (within $[0,\epsilon]$) to be strictly increasing in the welfare realized in each period.

Our negative results without cheap-talk communication, in contrast to the positive insights with cheap-talk communication, underscore the importance of requiring decision-makers to communicate with the public. Moreover, they highlight the type of information that is most valuable to the public: forecasts concerning the consequences of the actions actually taken.\footnote{Astute politicians frequently prefer to discuss the consequences of unobserved counterfactuals when communicating with the public.}

\paragraph{\textbf{Fixed Roles.}} In our setting, the agents' roles are fixed throughout the interaction. If, instead, both agents get the opportunity to serve as the DMs, then the natural test that picks the agent whose welfare performance is superior as the DM achieves both objectives; it identifies the more informed agent and incentivizes the DMs to take welfare-maximizing actions.

\subsection{Welfare Maximizing Agents}\label{sec:wel-max-agents}

In a non-strategic variant of our setting, the agents are unaffected by the test: the DM simply takes, and the AD advises on, the welfare-maximizing action. It is natural to ask whether an outside observer, who does not know the information structure $I$, can infer which agent is more informed from the sequence of observed outcomes. Intuitively, in an environment with one safe arm and one risky arm, whenever the DM takes the safe arm while the AD advises the risky one, the superiority of either decision cannot be inferred by an outside observer, as the outcome of the risky arm is not observed. The following example formalizes this intuition.

\begin{example}
 Consider a binary-arm setting with a safe arm~1 yielding a welfare of~$2$, and a risky arm~2 yielding a welfare of either~$1$ or~$3$. The uncertainty in this setting is captured by a single probability $p \in [0,1]$ representing the probability of the high realization~($3$) from arm~2. Let the prior be $p = \tfrac{8}{15}$.

An information structure can be represented as a two-stage martingale $(Y_1, Y_2)$ taking values in $\Delta([0,1])$, where $Y_2$ refines $Y_1$ in the sense of Blackwell. The distribution of $Y_1$ ($Y_2$) captures the beliefs of the less-informed (more-informed) agent.

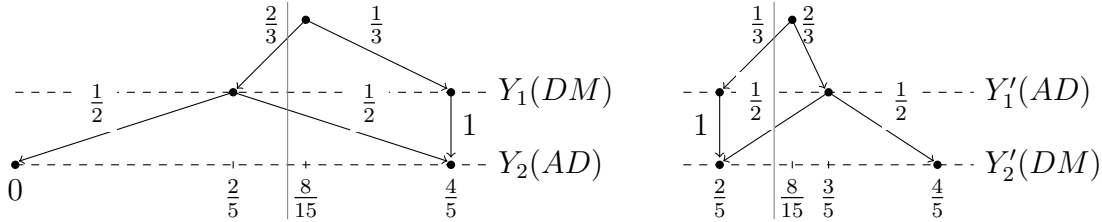
\begin{figure}[ht]
\begin{tikzpicture}[scale=0.48]

\draw[dashed]  (0,0)--(13,0) node[right] {$Y_2(AD)$};
\draw[dashed] (0,2)--(13,2) node[right] {$Y_1(DM)$};

\filldraw (8,4) circle(0.1);
\filldraw (6,2) circle(0.1);
\filldraw (12,2) circle(0.1);
\filldraw (0,0) circle(0.1);
\filldraw (12,0) circle(0.1);

\draw[->] (8,4) --  (6.1,2.1) {node[midway,above] () {$\frac{2}{3}$}};
\draw[->] (8,4) -- (11.9,2.1) {node[midway,above] () {$\frac{1}{3}$}};

\draw[->] (6,2) -- (0.1,0.1);
\node[fill=white] at (2.25,1.7) {$\frac{1}{2}$};
\draw[->] (6,2) -- (11.9,0.1);
\node[fill=white] at (9.75,1.7) {$\frac{1}{2}$};

\draw[->] (12,2) -- (12,0.2) {node[midway,right] () {$1$}};

\draw (0,0.1)--(0,-0.1) node[below] {$0$}; 
\draw (6,0.1)--(6,-0.1) node[below] {$\frac{2}{5}$}; 
\draw (8,0.1)--(8,-0.1) node[below] {$\frac{8}{15}$};
\draw (12,0.1)--(12,-0.1) node[below] {$\frac{4}{5}$};
\draw[gray] (7.5,4.5)--(7.5,-1.5);
\end{tikzpicture}
\hspace{5mm}
\begin{tikzpicture}[scale=0.48]

\draw[dashed]  (5,0)--(13,0) {node[right] {$Y'_2(DM)$}};
\draw[dashed] (5,2)--(13,2) {node[right] {$Y'_1(AD)$}};

\filldraw (8,4) circle(0.1);
\filldraw (6,2) circle(0.1);
\filldraw (9,2) circle(0.1);
\filldraw (6,0) circle(0.1);
\filldraw (12,0) circle(0.1);

\draw[->] (8,4) --  (6.1,2.1) {node[midway,above] () {$\frac{1}{3}$}};
\draw[->] (8,4) -- (8.9,2.1) {node[midway,above] () {$\frac{2}{3}$}};

\draw[->] (9,2) -- (6.1,0.1);
\node[fill=white] at (7,1.75) {$\frac{1}{2}$};
\draw[->] (9,2) -- (11.9,0.1);
\node[fill=white] at (11,1.75) {$\frac{1}{2}$};

\draw[->] (6,2) -- (6,0.2) {node[midway,left] () {$1$}};

\draw (6,0.1)--(6,-0.1) node[below] {$\frac{2}{5}$}; 
\draw (8,0.1)--(8,-0.1) node[below] {$\frac{8}{15}$};
\draw (9,0.1)--(9,-0.1) node[below] {$\frac{3}{5}$};
\draw (12,0.1)--(12,-0.1) node[below] {$\frac{4}{5}$};
\draw[gray] (7.5,4.5)--(7.5,-1.5);

\end{tikzpicture}
    \caption{The martingales $(Y_1, Y_2)$ and $(Y'_1, Y'_2)$ defining the information structures $I$ and $I'$. The gray line at the value $\tfrac{1}{2}$ represents the threshold above which arm~2 is optimal.}\label{fig:mar}
\end{figure}

Let $I$ be the information structure in which the AD is more informed, and let the joint distribution of the $(DM, AD)$’s beliefs $(Y_1, Y_2)$ take the values $(\tfrac{2}{5}, 0)$, $(\tfrac{2}{5}, \tfrac{4}{5})$, and $(\tfrac{4}{5}, \tfrac{4}{5})$, each with probability $\tfrac{1}{3}$.

Let $I'$ be the information structure in which the DM is more informed, and let the joint distribution of the $(DM, AD)$’s beliefs $(Y'_2, Y'_1)$ take the values $(\tfrac{2}{5}, \tfrac{2}{5})$, $(\tfrac{2}{5}, \tfrac{3}{5})$, and $(\tfrac{4}{5}, \tfrac{3}{5})$, each with probability $\tfrac{1}{3}$.

The two martingales, $Y$ and $Y'$, corresponding to these information structures are illustrated in Figure~\ref{fig:mar}.

In both information structures, $I$ and $I'$, if the agents act and advise according to their welfare-maximizing beliefs, the resulting distribution over outcomes $O$ is given by
\begin{align*}
    \mathbb{P}[o=(1,1,2)]=\mathbb{P}[o=(1,2,2)]=\frac{1}{3}, \mathbb{P}[o=(2,2,1)]=\frac{1}{15}, \text{ and } \mathbb{P}[o=(2,2,3)]=\frac{4}{15}.
\end{align*}
Hence, it is impossible to determine whether the underlying information structure is $I$ or $I'$, and therefore to identify the more-informed agent.
\end{example}
The impossibility of identifying the more-informed agent in a non-strategic environment provides a high-level intuition for a similar conclusion in the strategic setting, although the arguments in the latter case are more intricate. In light of this observation, the positive result of Theorem~\ref{theo:simul} may appear surprising. When agents do not attempt to manipulate the test, identification of the more-informed agent is impossible. Yet when both strategically manipulate the test, an effect that might be \emph{mistakenly} viewed as merely adding noise to the system, identification becomes possible.

\subsection{Future Directions}

Environments with unobserved counterfactuals are common in decision-making. It is often unclear which past problems are meaningfully comparable to the current one, making counterfactual estimation contentious. As a result, assessing the quality of decision-making in such settings is challenging—and sometimes impossible. Identifying the boundaries between what can and cannot be achieved in these environments is therefore an important research direction. This paper contributes to that effort by studying a special yet realistic scenario in which the reference point for evaluating the DM’s performance is an adviser (AD) who proposes alternatives to the current decision path. Developing other ordinal or cardinal approaches for assessing decision-making without observable counterfactuals remains an important avenue for future research.

In the specific context of this paper, and as discussed in Section~\ref{sec:generalTests}, it remains open to characterize the power of general tests beyond scoring tests. In particular, can such tests identify the more informed agent? And if so, can they also incentivize the DM to take welfare-improving actions?

Our analysis focuses on a homogeneous population with common interests. In many applications, however, preferences over decision outcomes may differ across groups. In such cases, it is natural to consider a variant of our framework in which outcomes are represented by a vector of group-specific utilities. In this environment, optimality is no longer well defined, and Pareto efficiency becomes a natural alternative. Some insights—especially negative ones—can be drawn from our results, which correspond to a single-dimensional special case of this broader setting.

Finally, our model assumes that decision consequences are observed immediately. In practice, decisions may have both short- and long-term effects. Incorporating gradual revelation of outcomes into models of decision assessment is an interesting direction for future work. Intuitively, delayed feedback only makes inference harder for an external observer, suggesting that our negative results are likely to extend to these more general environments.

\bibliographystyle{plainnat}
\bibliography{biblio}

@article{kremer2014implementing,
  title={Implementing the “wisdom of the crowd”},
  author={Kremer, Ilan and Mansour, Yishay and Perry, Motty},
  journal={Journal of Political Economy},
  volume={122},
  number={5},
  pages={988--1012},
  year={2014},
  publisher={University of Chicago Press Chicago, IL}
}

@article{lai1985asymptotically,
  title={Asymptotically efficient adaptive allocation rules},
  author={Lai, Tze Leung and Robbins, Herbert},
  journal={Advances in applied mathematics},
  volume={6},
  number={1},
  pages={4--22},
  year={1985},
  publisher={Academic Press}
}

@article{garivier2008upper,
  title={On upper-confidence bound policies for non-stationary bandit problems},
  author={Garivier, Aur{\'e}lien and Moulines, Eric},
  journal={arXiv preprint arXiv:0805.3415},
  year={2008}
}

@article{thompson1933likelihood,
  title={On the likelihood that one unknown probability exceeds another in view of the evidence of two samples},
  author={Thompson, William R},
  journal={Biometrika},
  volume={25},
  number={3/4},
  pages={285--294},
  year={1933},
  publisher={JSTOR}
}

@article{bolton1999strategic,
  title={Strategic experimentation},
  author={Bolton, Patrick and Harris, Christopher},
  journal={Econometrica},
  volume={67},
  number={2},
  pages={349--374},
  year={1999},
  publisher={Wiley Online Library}
}

@article{robbins1952some,
  title={Some aspects of the sequential design of experiments},
  author={Robbins, Herbert},
  journal={Bulletin of the American Mathematical Society},
  volume={58},
  number={5},
  pages={527–-535},
  year={1952}
}

@article{bikhchandani2024information,
  title={Information cascades and social learning},
  author={Bikhchandani, Sushil and Hirshleifer, David and Tamuz, Omer and Welch, Ivo},
  journal={Journal of Economic Literature},
  volume={62},
  number={3},
  pages={1040--1093},
  year={2024},
  publisher={American Economic Association 2014 Broadway, Suite 305, Nashville, TN 37203-2425}
}

@article{bikhchandani1992theory,
  title={A theory of fads, fashion, custom, and cultural change as informational cascades},
  author={Bikhchandani, Sushil and Hirshleifer, David and Welch, Ivo},
  journal={Journal of political Economy},
  volume={100},
  number={5},
  pages={992--1026},
  year={1992},
  publisher={The University of Chicago Press}
}

@article{banerjee1992simple,
  title={A simple model of herd behavior},
  author={Banerjee, Abhijit V},
  journal={The quarterly journal of economics},
  volume={107},
  number={3},
  pages={797--817},
  year={1992},
  publisher={MIT Press}
}

@article{gneiting2007strictly,
  title={Strictly proper scoring rules, prediction, and estimation},
  author={Gneiting, Tilmann and Raftery, Adrian E},
  journal={Journal of the American statistical Association},
  volume={102},
  number={477},
  pages={359--378},
  year={2007},
  publisher={Taylor \& Francis}
}

@article{flesch2023equilibrium,
  title={Equilibrium in two-player stochastic games with shift-invariant payoffs},
  author={Flesch, J{\'a}nos and Solan, Eilon},
  journal={Journal de Math{\'e}matiques Pures et Appliqu{\'e}es},
  volume={179},
  pages={68--122},
  year={2023},
  publisher={Elsevier}
}

@article{blackwell1956analog,
  title={An analog of the minimax theorem for vector payoffs.},
  author={Blackwell, David},
  journal={Pacific Journal of Mathematics},
  volume={6},
  number={1},
  pages={1--8},
  year={1956}
}

@book{kechris2012classical,
  title={Classical descriptive set theory},
  author={Kechris, Alexander},
  volume={156},
  year={2012},
  publisher={Springer Science \& Business Media}
}

@article{ShpitserPearl2012,
  author    = {Ilya Shpitser and Judea Pearl},
  title     = {What Counterfactuals Can Be Tested},
  journal   = {Journal of Machine Learning Research},
  year      = {2012},
  volume    = {13},
  pages     = {837--850},
}

@inproceedings{BalkePearl1994,
  author    = {Alexander Balke and Judea Pearl},
  title     = {Counterfactual Probabilities: Computational Methods, Bounds and Applications},
  booktitle = {Proceedings of the Tenth Conference on Uncertainty in Artificial Intelligence (UAI '94)},
  year      = {1994},
  pages     = {46--54},
  publisher = {Morgan Kaufmann},
  address   = {San Francisco, CA},
}

@book{Pearl2009Causality,
  author    = {Judea Pearl},
  title     = {Causality: Models, Reasoning, and Inference},
  edition   = {2},
  publisher = {Cambridge University Press},
  address   = {Cambridge},
  year      = {2009},
}

@article{Sandroni2003Reproducible,
  author    = {Alvaro Sandroni},
  title     = {The Reproducible Properties of Correct Forecasts},
  journal   = {International Journal of Game Theory},
  year      = {2003},
  volume    = {32},
  number    = {1},
  pages     = {151--159},
}

@article{OlszewskiSandroni2008Econometrica,
  author    = {Wojciech Olszewski and Alvaro Sandroni},
  title     = {Manipulability of Future-Independent Tests},
  journal   = {Econometrica},
  year      = {2008},
  volume    = {76},
  number    = {6},
  pages     = {1437--1466},
}

@article{OlszewskiSandroni2009AOS_Manipulability,
  author    = {Wojciech Olszewski and Alvaro Sandroni},
  title     = {Manipulability of Future-Independent Tests},
  journal   = {Annals of Statistics},
  year      = {2009},
  volume    = {37},
  number    = {2},
  pages     = {1010--1033},
}

@article{OlszewskiSandroni2009AOS_Nonmanipulable,
  author    = {Wojciech Olszewski and Alvaro Sandroni},
  title     = {A Nonmanipulable Test},
  journal   = {Annals of Statistics},
  year      = {2009},
  volume    = {37},
  number    = {2},
  pages     = {1013--1039},
}

@article{DekelFeinberg2006,
  author    = {Eddie Dekel and Yossi Feinberg},
  title     = {Non-Bayesian Testing of a Stochastic Prediction},
  journal   = {Review of Economic Studies},
  year      = {2006},
  volume    = {73},
  number    = {4},
  pages     = {893--906},
}

@article{SandroniShmaya2013,
  author    = {Alvaro Sandroni and Eran Shmaya},
  title     = {Nonmanipulable Prequential Tests for Exchangeable Theories},
  journal   = {Econometrica},
  year      = {2013},
  volume    = {81},
  number    = {6},
  pages     = {315--318},
}

@article{AlNajjarWeinstein2008,
  author    = {Nabil I. Al{-}Najjar and Jonathan Weinstein},
  title     = {Comparative Testing of Experts},
  journal   = {Econometrica},
  year      = {2008},
  volume    = {76},
  number    = {3},
  pages     = {541--559},
}

@article{KavalerSmorodinsky2019,
  author    = {Itay Kavaler and Rann Smorodinsky},
  title     = {On Comparison of Experts},
  journal   = {Games and Economic Behavior},
  year      = {2019},
  volume    = {118},
  pages     = {94--109},
}

@incollection{Olszewski2011Calibration,
  author    = {Wojciech Olszewski},
  title     = {Calibration and Expert Testing},
  booktitle = {The New Palgrave Dictionary of Economics},
  editor    = {Steven N. Durlauf and Lawrence E. Blume},
  edition   = {2},
  publisher = {Palgrave Macmillan},
  address   = {London},
  year      = {2011},
}

@article{blackwell1953equivalent,
  title={Equivalent comparisons of experiments},
  author={Blackwell, David},
  journal={The annals of mathematical statistics},
  pages={265--272},
  year={1953},
  publisher={JSTOR}
}

\appendix

\section{Proofs of the Results in Section \ref{sec:results}}\label{ap:proofs}

\subsection{Proof of Theorem \ref{theo:simul}}\label{sec:sim-proof}
We follow the approach outlined in Section~\ref{sec:intuition}, but without imposing Assumptions~\ref{as:unique} and~\ref{as:posteriors}.

Consider the scoring function $s$ defined in Equation~\eqref{eq:scor}. Recall the complete-information game $G(s,p)$ induced by the scoring function $s$. Using the notation $e_i = \mathbb{E}_p[X_i]$, the game $G(s,p)$ is specified as follows.
\begin{table}[ht]
\begin{center}
\begin{tabular}{cccccc}
                              & 1                                       & 2                                       & 3                                       & $\cdots$                      & $n$                                     \\ \cline{2-6} 
\multicolumn{1}{c|}{1}        & \multicolumn{1}{c|}{$e_1$}              & \multicolumn{1}{c|}{$-\frac{e_1}{n-1}$} & \multicolumn{1}{c|}{$-\frac{e_1}{n-1}$} & \multicolumn{1}{c|}{$\cdots$} & \multicolumn{1}{c|}{$-\frac{e_1}{n-1}$} \\ \cline{2-6} 
\multicolumn{1}{c|}{2}        & \multicolumn{1}{c|}{$-\frac{e_2}{n-1}$} & \multicolumn{1}{c|}{$e_2$}              & \multicolumn{1}{c|}{$-\frac{e_2}{n-1}$} & \multicolumn{1}{c|}{$\cdots$} & \multicolumn{1}{c|}{$-\frac{e_2}{n-1}$} \\ \cline{2-6} 
\multicolumn{1}{c|}{3}        & \multicolumn{1}{c|}{$-\frac{e_3}{n-1}$} & \multicolumn{1}{c|}{$-\frac{e_3}{n-1}$} & \multicolumn{1}{c|}{$e_3$}              & \multicolumn{1}{c|}{$\cdots$} & \multicolumn{1}{c|}{$-\frac{e_3}{n-1}$} \\ \cline{2-6} 
\multicolumn{1}{c|}{$\vdots$} & \multicolumn{1}{c|}{$\vdots$}           & \multicolumn{1}{c|}{$\vdots$}           & \multicolumn{1}{c|}{$\vdots$}           & \multicolumn{1}{c|}{$\ddots$} & \multicolumn{1}{c|}{$\vdots$}           \\ \cline{2-6} 
\multicolumn{1}{c|}{$n$}      & \multicolumn{1}{c|}{$-\frac{e_n}{n-1}$} & \multicolumn{1}{c|}{$-\frac{e_n}{n-1}$} & \multicolumn{1}{c|}{$-\frac{e_n}{n-1}$} & \multicolumn{1}{c|}{$\cdots$} & \multicolumn{1}{c|}{$e_n$}              \\ \cline{2-6} 
\end{tabular}
\end{center}
\caption{The game $G(s,p)$}\label{fig:game}
\end{table}

We first establish the uniqueness of the DM’s optimal strategy in this game, thereby validating Assumption~\ref{as:unique} for the scoring function~$s$.

\begin{lemma}\label{lem:opt}
The unique optimal strategy of the DM (the row player) in the game~$G(s,p)$ is given by $(\frac{c}{e_1},...,\frac{c}{e_n})$ where $c=(\frac{1}{e_1}+...+\frac{1}{e_n})^{-1}$.
\end{lemma}

\begin{proof}
    If, in equilibrium, the mixed strategy of the DM is not fully supported, then the AD can choose an action that the DM does not play and obtain a strictly negative score, which is impossible.

Furthermore, if, in equilibrium, the strategy of the AD is not fully supported, then the DM will avoid playing actions that the AD never uses, contradicting the requirement that the DM’s mixed strategy be fully supported.
    
    Let $(\alpha_1, \ldots, \alpha_n)$ be an optimal strategy for the DM. By the indifference principle, for every pair of actions $i, i' \in [n]$ of the AD, we have
    $$\frac{n}{n-1}\alpha_i e_i-\frac{1}{n-1}\sum_{j\in [n]} \alpha_j e_j=\frac{n}{n-1}\alpha_i' e_i-\frac{1}{n-1}\sum_{j\in [n]} \alpha_j e_j.$$
    Hence, $\alpha_i e_i = \alpha_{i'} e_{i'}$ for every $i, i' \in [n]$. This condition uniquely determines the strategy $(\alpha_1,...,\alpha_n)=(\frac{c}{e_1},...,\frac{c}{e_n})$ where $c=(\frac{1}{e_1}+...+\frac{1}{e_n})^{-1}$. 
\end{proof}

We now prove Theorem~\ref{theo:simul}. By Lemma~\ref{lem1}, it suffices to show that $\val[G(s,I)] \ge 0$ if and only if the DM is strictly more informed.

Suppose first that the DM is strictly more informed. Consider the zero-sum game in which the DM ignores her additional information and bases her decision solely on $q^j$. This game, denoted $G(s, q^j)$ and depicted in Table~\ref{fig:game}, admits the mixed strategy $(\tfrac{c}{e_1}, \ldots, \tfrac{c}{e_n})$ for $c = \left(\tfrac{1}{e_1} + \cdots + \tfrac{1}{e_n}\right)^{-1}$, which guarantees a score of~0 for the DM. Therefore, when the DM utilizes her additional information~$\br$, she can only achieve a (weakly) higher score in the actual game with information $I = (p, \bq, \br)$, and hence $\val[G(s,I)] \ge 0$.

Now consider the case in which the AD is strictly more informed. This implies that there exists some $q^j$, without loss of generality, $j = 1$, such that $\{r^{1, j'} : j' \in [l]\} \not\subset A_i$ for every $i \in [n]$; that is, there is no single welfare-maximizing action common to all $r^{1, j'}$ with $j' \in [l]$. 

By the Minimax Theorem, it suffices to show that for every strategy of the DM, the AD can best respond and obtain a strictly negative score. Let $\alpha = (\alpha_1, \ldots, \alpha_n)$ be a mixed strategy of the DM when she receives signal $j = 1$. Denote by $D = \arg\min {\alpha_i : i \in [n]} \subset [n]$ the set of actions played with minimal probability under $\alpha$, and recall that $N^*(r^{1, j'})$ denotes the set of welfare-maximizing actions at the posterior $r^{1, j'}$. The AD will construct a best response that exploits these minimal-probability actions to get a strictly negative score.

We know that there exists some $j' \in [l]$ such that $D \neq N^*(r^{1, j'})$, without loss of generality, let $j' = 1$. By Lemma~\ref{lem:opt}, in the game $G(s, r^{1,1})$, the DM must assign the minimal probability precisely to the actions in $N^*(r^{1, j'})$. Equivalently, after the AD receives the signal $(1,1)$, the DM does not play an optimal strategy, allowing the AD to achieve a strictly negative expected score. For all other signals $(j, j') \neq (1,1)$, the AD can play uniformly and guarantee an expected score of $0$. Overall, such a strategy yields the AD a strictly negative total expected score. This concludes the proof of the theorem.

\subsection{Proof of Corollary \ref{cor:finite}}\label{ap:proof-cor}
The proof closely follows the argument used in the proof of Lemma~\ref{lem1} in Appendix~\ref{ap:missing}, with the only difference being that, instead of applying the law of large numbers, we use Hoeffding’s inequality. Note that the scoring function is uniformly bounded by $K$, i.e., $|s(o[t])| \le K$.

If the DM is more informed, she can use her optimal strategy in the game $G(s, I[t])$ and guarantee an expected score of~0 in every period $t \in \mathbb{N}$. By Hoeffding’s inequality, we obtain
\begin{align*}
    \mathbb{P}\left[ \sum_{t\in [T]} s(o[t])\leq -T^{3/4} \right] \leq \exp \left(-\frac{\sqrt{T}}{2K} \right).
\end{align*}

    If the AD is more informed, then by the definition of strict information superiority (see Definition~\ref{def:dynamic-strict}), there exists $\epsilon > 0$ such that the AD is $\epsilon$-strictly more informed (i.e., Equation~\eqref{eq:epsilon-strict} holds) for all periods $t \in \mathbb{N}$. Let $\mathcal{I}_\epsilon$ denote the set of all information structures satisfying Equation~\eqref{eq:epsilon-strict}. Notice that $\mathcal{I}_\epsilon$ is a closed set and that $\val[G(s, I)]$ is a continuous function of $I$. Therefore, its maximum is attained and denoted by $-\delta < 0$. If the AD uses his optimal strategy in the game $G(s, I[t])$, he guarantees a score of at most $-\delta$ in every period $t \in \mathbb{N}$. By Hoeffding’s inequality, we obtain
 \begin{align*}
    &\mathbb{P}\left[ \sum_{t\in [T]} s(o[t])\geq -T^{3/4} \right] = \mathbb{P}\left[ \sum_{t\in [T]} s(o[t]) + \delta T \geq -T^{3/4} +\delta T \right] \leq \\
    \leq &\mathbb{P}\left[ \sum_{t\in [T]} s(o[t])+\delta T \geq T^{3/4} \right] \leq  \exp \left(-\frac{\sqrt{T}}{2K} \right) 
\end{align*}
where the first inequality holds for all sufficiently large $T$, specifically for $T \ge 16\delta^{-4}$.

\subsection{Proof of Theorem \ref{theo:welfare}}\label{sec:welfare-proof}

Consider the following instance with two arms. Arm~1 is a safe arm that deterministically yields a welfare of $k \gg 1$, while Arm~2 is a risky arm that yields either a low welfare of~1 or a high welfare of~$k^3$. The only source of uncertainty lies in the probability~$p$ that the risky arm produces the high welfare~$k^3$. In this case, a scoring function is fully characterized by six numbers: $s(1, i, k)$ for the two actions $i = 1, 2$ of the AD, and $s(2, i, u)$ for $i = 1, 2$ and $u \in \{1, k^3\}$.

In the game $G(s,p)$, we introduce the notation for maximal welfare, analogous to 
$U^*_{DM}$ in Equations~\eqref{eq:udm1} and~\eqref{eq:udm2}, and for equilibrium 
welfare, analogous to $W(s)$ in Equation~\eqref{eq:w}. The maximal welfare is simply $u^*(p)$, while the welfare in equilibrium is given by

\begin{align*}
  w(s,p)=\min \{ \E_{\eta,p} [u_{a_{DM}}]: \eta \text{ is and equilibrium of } G(s,p) \}. 
\end{align*}
The ratio of these two terms is denoted by $R(p)=w(s,p)/u^*(p)$.

We argue that it suffices to show that, for every separating scoring function $s$ and every $c' > \tfrac{1}{2}$, there exists a prior $p$ such that $R(p) \le c'$. Indeed, assuming this claim holds, given any separating scoring function $s$ and a constant $c > \tfrac{1}{2}$, we can choose some $c' \in (\tfrac{1}{2}, c)$ and a prior $p$ for which $R(p) \le c'$.

Now, consider the information structure $I = (p, \bq, \br)$ defined as follows: with probability $(1 - \epsilon)$, both the DM and the AD receive the posterior $q^1 = r^{1,1} = p$; with probability $\epsilon$, the AD receives the posterior $q^2 = p' \neq p$ for some $0 < p' < 1$, while the DM observes the realization of the second arm (drawn from the prior $p'$).

With probability $(1 - \epsilon)$, the agents play the game $G(s, p)$. Therefore, the total expected welfare satisfies 
\begin{equation*}
    W(s,I)\leq (1 - \epsilon)\, \mathbb{E}_{\eta, p}[u_{a_{DM}}] + \epsilon k^3 
        \le c' u^*(p) + \epsilon k^3 
        \le (c' + \epsilon k^3) u^*(p).
\end{equation*}
For sufficiently small~$\epsilon$, we have $c' + \epsilon k^3 \le c$. Hence, the information structure $I$ constitutes an example of strict informational superiority for which the welfare does not exceed $cu^*(I)$.

Henceforth, we prove the existence of a bad prior $p$ for which $R(p) \le c$, given any separating scoring function. Our arguments apply to both cases, whether the tie-breaking agent is the DM or the AD.

We begin by showing that the signs of the scores must follow the pattern of a matching-pennies game in both $G(s,0)$ and $G(s,1)$—that is, in the games where the realization of arm~2 is known to both agents. Roughly speaking, we establish that
\begin{align*}
&\sign(s(1,1,k))=\sign(s(2,2,1))\neq \sign(s(1,2,k))=\sign(s(2,1,1)) \text{ and } \\
&\sign(s(1,1,k))=\sign(s(2,2,k^3))\neq \sign(s(1,2,k))=\sign(s(2,1,k^3)),    
\end{align*}
while treating separately the cases in which $s(i, j, u) = 0$.

We first focus on the signs of the pair $(s(1,1,k), s(1,2,k))$, which, for brevity, we denote by $(x, y)$. It is impossible to have $x, y \le 0$, since in the game $G(s,1)$ (i.e., the high realization of arm~2), the equilibrium (with game value~0) would have the DM choosing arm~1 with probability~1, yielding an approximation ratio of $R(1) = \tfrac{1}{k^2} < c$. Similarly, it is impossible to have $x, y \ge 0$, since in the game $G(s,0)$ (i.e., the low realization of arm~2), the equilibrium (also with game value~0) would have the DM choosing arm~2 with probability~1, yielding an approximation ratio of $R(1) = \tfrac{1}{k} < c$. Therefore, the only remaining possibility is that $\sign(x) \ne \sign(y)$. Without loss of generality (by relabeling the AD’s actions), we assume $x > 0$ and $y < 0$. Moreover, we normalize the scoring function by scaling it so that $x = s(1,1,k) = 1$, which does not affect the analysis.

We now focus on the game $G(s,0)$ (the low realization of arm~2) and show that $s(2,1,1) < 0$ and $s(2,2,1) > 0$.

First, $s(2,2,1) \le 0$ is impossible: in that case, the zero-value equilibrium of this game would necessarily be the pure equilibrium $(a_{DM}, a_{AD}) = (2,2)$, implying that the DM chooses the risky arm with probability~1 and achieves an approximation ratio of $R(0) = \tfrac{1}{k} < c$. Hence, $s(2,2,1) > 0$.

Second, $s(2,1,1) \ge 0$ is also impossible, since in that case the action $a_{DM} = 2$ would guarantee a (weakly) positive score and would thus be chosen in an inefficient equilibrium. Therefore, we must have $s(2,1,1) < 0$.

We now focus on the game $G(s,1)$ (the high realization of arm~2) and show that \\ $s(2,1,k^3) \le 0$ and $s(2,2,k^3) \ge 0$.

First, $s(2,2,k^3) < 0$ is impossible, since in that case the action $a_{AD} = 2$ would guarantee a strictly negative score; hence, $s(2,2,k^3) \ge 0$.

Assume, by way of contradiction, that $s(2,1,k^3) > 0$. Consider the game $G(s,1 - \epsilon)$, where the scores for the action $a_{DM} = 2$ are 
$$((1-\epsilon)s(2,1,k^3)+\epsilon s(2,1,1),(1-\epsilon)s(2,2,k^3)+\epsilon s(2,2,1))$$
for the actions $a_{AD} = 1, 2$, respectively. For sufficiently small $\epsilon$, both expressions are strictly positive. This implies that in the game $G(s,1 - \epsilon)$, the DM can guarantee a strictly positive score by playing $a_{DM} = 2$, leading to a contradiction.

To summarize, we have shown that for $p < 1$, the game $G(s,p)$ exhibits the following pattern of score signs: $s_p(1,1), s_p(2,2) > 0$ and $s_p(1,2), s_p(2,1) < 0$, where $s_p : [2] \times [2] \to \mathbb{R}$ denotes the expected score under prior $p$ (this notation is relevant for $s_p(2, \cdot)$ only). Moreover, we have normalized $s_p(1,1) = 1$.

The game $G(s,p)$ then has a unique mixed equilibrium in which the DM chooses action~2 with probability $\frac{1}{1 - s_p(2,1)}$. Observe that $-s_p(2,1)$ is a linear function of $p$, which we denote by $L(p) = -s_p(2,1)$.\footnote{The negative sign is used so that $L(p)$ is positive.} We also have $L(p) \ge 0$ for all $p \in [0,1]$.

We now turn to the welfare-efficiency analysis of this mixed behavior, in which the risky arm is chosen with probability $\frac{1}{1 + L(p)}$, for a specific (carefully chosen) value of $p = \tfrac{1}{k}$. In this case, the risky arm (arm~2) is superior and yields an expected welfare of $k^2$, which is substantially higher than the safe welfare of $k$. The resulting approximation ratio is given by
\begin{align*}
    R(p)=\frac{\frac{L(p)}{1+L(p)}k + \frac{1}{1+L(p)}k^2}{k^2} = \frac{\frac{1}{1+L(p)}k^2+O(k)}{k^2}=\frac{1}{1+L(p)}+O\left(\frac{1}{k} \right).
\end{align*}

If $R(p) \le c$, we have found a prior $p$ as required. Otherwise, suppose that $R(p) \ge c > \tfrac{1}{2}$. Then, by the equation above, we must have $L(p) \le 1$. Since $L(p)$ is linear, positive, and satisfies $L(\tfrac{1}{k}) \le 1$, its slope cannot be more negative than $-1 - O(\tfrac{1}{k})$. Consequently, $L(0) \le 1 + O(\tfrac{1}{k})$.

This implies that in the game $G(s,0)$, where the safe action is welfare-maximizing and yields a welfare of $k$ (which is substantially higher than~$1$), the safe action is taken only with probability
$$1-\frac{1}{1+L(p)} \leq 1-\frac{1}{2+O(\frac{1}{k})}=\frac{1}{2}+O \left( \frac{1}{k} \right)$$ and the approximation ratio is at most
$$R(0)=\frac{(\frac{1}{2}+O(\frac{1}{k}))k+(\frac{1}{2}-O(\frac{1}{k}))\cdot 1}{k}=\frac{1}{2}+O \left(\frac{1}{k} \right)<c,$$
Hence, we have found a prior $p = 0$ as required. This concludes the proof of Theorem~\ref{theo:welfare}.

\subsection{Proof of Theorem \ref{theo:seq}}\label{sec:seq-proof}

By Lemma~\ref{lem2}, it suffices to prove the nonexistence of a separating scoring function $s$. We establish this through a unified line of argument that applies to both cases: when the DM acts first and when the AD acts first. We refer to the first-moving agent as the \emph{leader} (LE) and to the second-moving agent as the \emph{follower} (FO), where $\{LE, FO\} = \{DM, AD\}$. We assume that the FO is the minimizing agent (after multiplying all scores by $-1$ when $FO = DM$).

We consider a binary-arm environment ($n = 2$). The set of possible realizations of the arms (the set $U$) will be specified later. 

For $i = 1, 2$, we consider the (hypothetical) zero-sum game $H_i$ between \emph{nature} (the maximizing agent) and the FO (the minimizing agent). Nature chooses a realization profile of the arms $(u_1, u_2) \in U$, while the FO simultaneously chooses an arm $j \in [2]$. The payoff is determined by the score $s$ of the resulting outcome $o = (i, j, u_k)$, where $k$ denotes the action taken by the DM (either $i$ or $j$).

We denote by $c_i = \val[H_i]$ the value of the game $H_i$, and by $Z_i \subset \Delta(U)$ the set of nature’s optimal (mixed) strategies in $H_i$. Let $y_i \in \Delta([2])$ be an optimal strategy of the FO in $H_i$. Finally, for $p \in \Delta(U)$, we denote by $\best_i(p)$ the best-reply payoff of the FO against nature’s mixed action $p$ in the game $H_i$.

We begin by showing that any separating scoring function $s$ must satisfy that the values of both games $H_1$ and $H_2$ are zero; that is, $c_1 = c_2 = 0$.

For $i = 1, 2$, we cannot have $c_i > 0$, since in that case, for any $p \in Z_i$, the leader in the game $G(s,p)$ could choose action $i$ and thereby secure a strictly positive score. Therefore, $c_1, c_2 \le 0$.

Assume that $c_1 < 0$. Then, in the game $H_2$, for every prior $p \in \Delta(U)$, we must have $\best_2(p) \ge 0$. Otherwise, in the game $G(s,p)$, the FO could guarantee a strictly negative score by playing $y_1$ after the LE chooses action~1 and by playing a best-reply action that yields $\best_2(p) < 0$ after the LE chooses action~2.
On the other hand, since $c_2 \le 0$, we must also have $\best_2(p) \le 0$, and therefore $\best_2(p) = 0$ for every prior $p$. If this holds, then in every game $G(s,I)$ with more-informed and less-informed agents, both agents can guarantee a score of~0: the leader by choosing action~2, and the follower by playing $y_1$ after action~1 and a best reply to his belief about the state after action~2. Hence, $\val[G(s,I)] = 0$ for every $I$, implying that $s$ is not separating. Therefore, $c_1 = 0$. By applying a symmetric argument, we similarly obtain $c_2 = 0$.

The FO cannot serve as the tie-breaking agent, since the action $y_i$ following the leader’s choice of $i = 1, 2$ guarantees an expected score of~0 regardless of the realizations of the arms. This remains true even if the FO is less informed—or entirely uninformed. Therefore, the tie-breaking agent must be the LE.

We now show that the sets $Z_1$ and $Z_2$ must take a very specific form. Recall that $Z_i$ is the set of all priors that guarantee a score of~0 for the leader, whereas any prior $p \in \Delta(U) \setminus Z_i$ yields a strictly negative score if the FO knows this prior (i.e., if $\best_i(p) < 0$). 

First, we argue that $Z_1 \cup Z_2 = \Delta(U)$. Indeed, if $p \notin Z_1 \cup Z_2$, then in the game $G(s,p)$ the FO obtains a strictly negative score, since $\max_{i = 1, 2} \best_i(p) < 0$.

Second, we have $Z_i \subset A_j$ for some $j \in {1, 2}$, recalling that $A_j$ is the set of priors for which the welfare-maximizing action is $j$. Suppose, by way of contradiction, that this is not the case. Then there exist $z \in \inter(A_1)$ and $z' \in \inter(A_2)$ such that $z, z' \in Z_i$. We can now repeat the argument used in the proof of Lemma~\ref{lem:equiv}. Consider an information structure in which the FO is strictly more informed and receives the posteriors $z$ and $z'$ with equal probability $\tfrac{1}{2}$, while the LE is uninformed. In this case, the LE can guarantee a score of~0, and be selected, since the LE is the tie-breaking agent, by choosing action~$i$, even though she is strictly less informed.

To summarize, the two convex closed sets $Z_1$ and $Z_2$ satisfy $Z_1 \cup Z_2 = \Delta(U) = A_1 \cup A_2$, and each $Z_i$ is contained in some $A_j$. This is possible only in one of the following two cases: either $Z_1 = A_1$ and $Z_2 = A_2$, or vice versa, $Z_1 = A_2$ and $Z_2 = A_1$.

We now arrive at the final step of the proof, which is to show that $Z_1 = A_1$ is impossible (and, by symmetry, $Z_1 = A_2$ is also impossible). To this end, we use the fact that the scoring function cannot depend on counterfactuals; that is, the score may depend only on the realization of one of the arms.

We now specify the set of realizations $U$. Assume that both arms $i = 1, 2$ have a low realization of~1 and a high realization of~2; that is, $U_1 = U_2 = \{1, 2\}$. A prior $p$ is then given by $p = (p(1,1), p(1,2), p(2,1), p(2,2))$. In this setting, $A_1=\{p:p(2,1)\geq p(1,2)\}$. 

Let $L_i : \Delta(U) \to \mathbb{R}$ denote the linear function representing the FO’s score when he chooses action $i \in \{1, 2\}$ in the game $H_1$. The assumption that $Z_1 = A_1$ implies that one of these two linear functions must satisfy $L_i \equiv 0$. Without loss of generality, let $L_1 \equiv 0$. The second linear function, $L_2$, must then be strictly positive exactly on the set $\Delta(U) \setminus A_1 = \{p : p(2,1) > p(1,2)\}$.

However, $L_2$ depends only on the realization of a single arm. Without loss of generality, assume it depends on the realization of Arm~1. Thus, $L_2(p)$ has the form
$$L_2(p)=\alpha (p(1,1)+p(1,2))+\beta (p(2,1)+p(2,2))+\gamma \text{ form some constants } \alpha,\beta,\gamma.$$
The following two equations must therefore be equivalent.
$$L_2(p)=0 \Leftrightarrow p(1,2)=p(2,1).$$
But this is impossible. Indeed, the fact that $p(1,1)$ does not appear in the equation $p(1,2) = p(2,1)$ implies that $\alpha = 0$, while the fact that $p(2,2)$ does not appear in the same equation implies that $\beta = 0$. However, $\alpha = \beta = 0$ implies that the set of solutions to $L_2(p) = 0$ is either empty (if $\gamma \neq 0$) or the entire simplex $\Delta(U)$ (if $\gamma = 0$). This contradiction completes the proof of Theorem~\ref{theo:seq}. $\blacksquare$

\subsection{Proof of Proposition \ref{pro:ad1}}
By Lemma \ref{lem1}, it suffices to show that $\val[G(s,I)]>0$ if and only if the DM is strictly more informed.

We say that the AD’s advice is welfare-maximizing if $a_{AD}\in \arg\max_{i=1,2}\E[X_i]$, recalling that $\E[X_1]$ is the deterministic payoff of the safe action.

Suppose first that the DM is strictly more informed. Consider the following strategy for the DM. Using her additional information, the DM determines whether the AD’s advice is welfare-maximizing. If it is, she chooses the safe action $a_{DM}=1$. If it is welfare-minimizing, she instead chooses the risky action $a_{DM}=2$. Since the DM is strictly more informed with strictly positive probability, she chooses the risky action with strictly positive probability; in this event, the expected score is strictly positive. In all remaining cases, the score is zero. Hence, the expected score is strictly positive.

Now suppose that the AD is strictly more informed. Consider the strategy in which the AD always provides welfare-maximizing advice. Since the DM has no additional information, she knows that the advice is welfare-maximizing. Choosing $a_{DM}=1$ then yields a score of zero, while choosing 
$a_{DM}=2$ yields a weakly negative score (because the advice is welfare-maximizing). Thus, the AD guarantees a weakly nonpositive score in the game. This completes the proof of Proposition \ref{pro:ad1}.

\subsection{Proof of Proposition \ref{pro:dm1}}
Consider the following instance with two arms. Arm 1 is a safe arm that yields a deterministic welfare of 2, while Arm 2 is a risky arm that yields either a low welfare of 1 or a high welfare of 3. The only source of uncertainty is the probability $p$ that the risky arm yields the high welfare 3. In this setting, a scoring function is fully characterized by six numbers: 
$s(1, i, 2)$ for the two actions $i = 1, 2$ of the AD, and $s(2, i, u)$ for $i = 1, 2$ and $u \in \{1, 3\}$.

We claim that $\min\{s(1,1,2),s(1,2,2)\}\geq 0$, that is, the safe arm must induce a weakly nonnegative score after the score-minimizing action of the AD. Suppose instead that \\ 
$\min\{s(1,1,2),s(1,2,2)\}< 0$. If there exists some pair $(i,u)$ such that 
$s(2,i,u)<0$, then in the game $G(s,p)$, where $p$ assigns probability one to 
$u$, the AD can guarantee a strictly negative score. Alternatively, if $\min\{s(1,1,2),s(1,2,2)\}< 0$ and $s(2,i,u)\geq 0$ for all pairs $(i,u)$, then the risky action becomes dominant for the DM, independently of the information structure $I$, which is impossible.

Next, recall from the proof of Theorem \ref{theo:seq} that the leader must be the tie-breaking agent. This argument applies to any binary-arm environment, and in particular to the setting of Proposition \ref{pro:dm1}. Since the DM is the leader in the present case, the inequality 
$\min\{s(1,1,2),s(1,2,2)\}\geq 0$ implies that the DM can guarantee $\val[G(s,I)]\geq 0$ by choosing the safe action, even when she is less informed. This contradiction completes the proof of Proposition \ref{pro:dm1}.




\section{Missing Proofs in Sction \ref{sec:intuition}}\label{ap:missing}

\begin{proof}[Proof of Lemma \ref{lem1}]
    Assume that the tie-breaking agent is the DM. If the DM is more informed, she can apply her optimal strategy in the game $G(s,I[t])$ in every period $t \in \mathbb{N}$. By the law of large numbers, this guarantees $\liminf_{T\to \infty} \frac{1}{T}\sum_{t\in [T]}s(o[t]) \geq 0$.

    The treatment of the case in which the AD is more informed is slightly more involved. By the definition of strict informational superiority (see Definition~\ref{def:dynamic-strict}), there exists $\epsilon > 0$ such that the DM is $\epsilon$-strictly more informed (i.e., Equation~\eqref{eq:epsilon-strict} holds) for all periods $t \in \mathbb{N}$. Let $\mathcal{I}_\epsilon$ denote the set of all information structures satisfying Equation~\eqref{eq:epsilon-strict}. Observe that $\mathcal{I}_\epsilon$ is a closed set, and that $\val[G(s,I)]$ is a continuous function of $I$. Hence, its maximum is attained, and we denote it by $-\delta < 0$.

If the AD plays his optimal strategy in the game $G(s,I[t])$ in every period $t \in \mathbb{N}$, then by the law of large numbers, this guarantees
$\liminf_{T\to \infty} \frac{1}{T}\sum_{t\in [T]}s(o[t]) \leq -\delta$.

    The case in which the tie-breaking agent is the AD is treated in a similar manner.
\end{proof}

\begin{proof}[Proof of Lemma \ref{lem2}] The nonexistence of a separating scoring function implies that, for every scoring function $s$ and every choice of tie-breaking agent, there exist an information structure $I$ and a less-informed agent $AG \in {DM, AD}$ with a strategy $\sigma \in \Delta([n])$ that guarantees the wrong sign of the expected score in the game $G(s,I)$ (or a wrong expected value of zero if $AG$ is the tie-breaking agent). For the sequence of information structures $I[t] = I$ (i.e., a fixed environment), by repeatedly playing $\sigma$ in every period, the agent $AG$ ensures, by the law of large numbers, that he is wrongly selected in the repeated game as well.  
\end{proof}

\section{Axiomatization of Scoring Tests}\label{ap:axiom}
A test maps the observed sequence of outcomes $\bo = (o[t])_{t \in \mathbb{N}} \in O^{\mathbb{N}}$—that is, the sequence of actions and welfare realizations of the DM’s actions—into a selection of either the DM or the AD. Roughly speaking, \emph{empirical measurability} requires that the test depend only on the limiting frequency of each outcome $o \in O$ (recall that $O$ is a finite set). To extend this definition to non-convergent sequences, one may fix a selection rule over the partial limits, such as the lexicographic selection.

Formally, for a finite $T \in \mathbb{N}$, let $\emp_T(\bo) \in \Delta(O)$ denote the empirical distribution of the first $T$ periods—that is, the uniform distribution over the outcomes $o[1], \ldots, o[T]$. Let $\Lim(\bo) \subset \Delta(O)$ denote the closed set of all partial limits of the sequence $(\emp_T(\bo))_{T \in \mathbb{N}}$.

We fix a selection function $\lex : 2^{\Delta(O)} \to \Delta(O)$ that selects the minimal element from any closed subset of $\Delta(O)$ according to a lexicographic order, given some fixed enumeration of the elements of $O$ (used to define the lexicographic hierarchy). We then define $L^*(\bo) = \lex(\Lim(\bo)) \in \Lim(\bo)$ as the \emph{limit of the empirical frequency} of the sequence $\bo$. In particular, if $\bo$ is a convergent sequence, then $\Lim(\bo)$ is a singleton, and $L^*(\bo)$ coincides with its limit.

We are now ready to define the property of empirical measurability.

\begin{definition}
A test $\theta : O^{\mathbb{N}} \to \{DM, AD\}$ is said to be \emph{empirically measurable} if, for every pair of observation sequences $\bo, \bo' \in O^{\mathbb{N}}$ satisfying $L^*(\bo) = L^*(\bo')$, we have $\theta(\bo) = \theta(\bo')$.
\end{definition}
That is, in an empirically measurable test, the limit frequency $L^*$ determines the winner. For such tests, we can equivalently view $\theta$ as a mapping $\theta : \Delta(O) \to {DM, AD}$.

The \emph{convexity} property further requires that the set of winning frequencies corresponding to each agent be convex.

\begin{definition}
An empirically measurable test $\theta : \Delta(O) \to \{DM, AD\}$ is said to be \emph{convex} if, for every two empirical frequencies $L^*_1, L^*_2 \in \Delta(O)$ such that $\theta(L^*_1) = \theta(L^*_2) = AG$, we have $\theta(\tfrac{1}{2}L^*_1 + \tfrac{1}{2}L^*_2) = AG$, for $AG \in \{DM, AD\}$.
\end{definition}

Convexity can be motivated by considering two convergent sequences $\bo$ and $\bo'$, and an alternating sequence $\bo''$ in which $\bo$ occurs in the odd periods and $\bo'$ occurs in the even periods. Notice that $L^*(\bo'') = \tfrac{1}{2}L^*_1 + \tfrac{1}{2}L^*_2$. If, by observing only the odd periods in $\bo''$, the conclusion is that $AG$ is the more-informed agent, and by observing only the even periods in $\bo''$, the conclusion is again that $AG$ is the more-informed agent, then it is natural to assume that observing the entire sequence $\bo''$ should lead to the same conclusion.

Empirical measurability, together with convexity, almost fully determines the test to be a scoring test, as stated in the following observation.

\begin{observation}\label{lem:sc}
    If a test $\theta$ is empirically measurable and convex, then there exists a scoring function $s:O \to \R$ such that 
    \begin{align*}
       s(L^*(\bo))>0 \Rightarrow \theta(L^*)=DM \ \text{ and } \ s(L^*(\bo))<0 \Rightarrow \theta(L^*)=AD. 
    \end{align*}
\end{observation}

\begin{proof}[Proof of Observation \ref{lem:sc}]
The two winning sets corresponding to the agents are disjoint convex sets whose union covers the simplex $\Delta(O)$. Consequently, each of these sets must be a half-space. Let $L : \Delta(O) \to \mathbb{R}$ denote the linear function that defines these half-spaces; that is, $L$ is positive on one of them and negative on the other. The scoring function $s(o)$ is then defined as the value of $L$ at the Dirac measure corresponding to the observation~$o$.
\end{proof}

This characterization is almost identical to the definition of a scoring test, except for the tie-breaking rule (i.e., the case $s(L^*(\bo)) = 0$). In scoring tests, the tie-breaking rule is required to select the same agent in all such cases, whereas empirically measurable and convex tests impose no such restriction. To precisely characterize the class of scoring tests, one can additionally require that the winning set of one of the agents be either open or closed.

\begin{observation}\label{lem:sc2}
    If a test $\theta$ is empirically measurable and convex, and if the set $W_{DM} = \{L^* : \theta(L^*) = DM\}$ is either open or closed, then $\theta$ is a scoring test.
\end{observation}

\begin{proof}[Proof of Observation \ref{lem:sc2}]
    The additional requirement that the winning set be open or closed enforces a consistent selection of the same agent along the boundary between the two winning sets (i.e., where $s(L^*(\bo)) = 0$).
\end{proof}
\end{document}